\documentclass[10pt,onecolumn,draftcls]{IEEEtran}
\def\comment#1{}

\usepackage[margin=1in]{geometry}

\usepackage{graphicx,verbatim}
\usepackage{amsmath}
\usepackage{amssymb,setspace}

\usepackage{xspace}
\usepackage{bbm}

\newcommand{\Ac}{\mathcal{A}}
\newcommand{\Bc}{\mathcal{B}}

\newcommand{\Ec}{\mathcal{E}}
\newcommand{\Fc}{\mathcal{F}}

\newcommand{\Nc}{\mathcal{N}}

\newcommand{\Qc}{\mathcal{Q}}

\newcommand{\Uc}{\mathcal{U}}

\newcommand{\ex}{{\rm e}}

\newcommand{\ev}{{\bf e}}

\newcommand{\xv}{{\bf x}}
\newcommand{\yv}{{\bf y}}
\newcommand{\zv}{{\bf z}}
\newcommand{\uv}{{\bf u}}
\newcommand{\vv}{{\bf v}}
\newcommand{\sv}{{\bf s}}

\def\a{\alpha}

\let\P\relax
\DeclareMathOperator\P{P}

\newcommand\ie{i.e.,\xspace}
\def\textiid{i.i.d.\@\xspace}
\newcommand\iid{\ifmmode\text{ i.i.d. } \else \textiid \fi}

\usepackage{caption}
\usepackage{subcaption}

\graphicspath{{figures/} }

\newcommand{\argmin}{\mathrm{argmin}}

\newcommand{\xvh}{\hat{\bf x}}
\newcommand{\uvh}{\hat{{\bf u}}}
\newcommand{\xvt}{\tilde{{\bf x}}}

\newcommand{\uvt}{\tilde{{\bf u}}}
\newcommand{\mvec}[1]{{\bf #1}}

\usepackage{cite}\usepackage{amsthm}\usepackage{dsfont}\usepackage{array}\usepackage{mathrsfs}\usepackage{comment}

\usepackage[noend]{algpseudocode}
\usepackage{algorithmicx}
\usepackage{algorithm}

\usepackage{hyperref}
\usepackage{breakurl}

\usepackage{color}
\usepackage[normalem]{ulem}

\begin{document}

\title{Solving  inverse problems via auto-encoders}

\author{
Pei Peng\thanks{P. Peng is a Ph.D. student at  Department of Electrical \& Computer Engineering,  Rutgers University, \texttt{pp566@scarletmail.rutgers.edu}}, 
Shirin Jalali\thanks{S. Jalali is with Nokia Bell Labs, 600 Mountain Avenue, Murray Hill, NJ, 07974, USA, \texttt{shirin.jalali@nokia-bell-labs.com}}
and 
Xin Yuan\thanks{X. Yuan is with Nokia Bell Labs, 600 Mountain Avenue, Murray Hill, NJ, 07974, USA, \texttt{xyuan@bell-labs.com} }
\footnote{This paper was  presented in part at 2019 IEEE International Symposium on Information Theory, Paris, France \cite{jalali2019solving} and will be presented in part at NeurIPS 2019 Workshop on ``Solving Inverse Problems with Deep Networks''.}
}

\maketitle

\doublespacing

\theoremstyle{plain}\newtheorem{lemma}{\textbf{Lemma}}\newtheorem{theorem}{\textbf{Theorem}}\newtheorem{corollary}{\textbf{Corollary}}\newtheorem{assumption}{\textbf{Assumption}}\newtheorem{example}{\textbf{Example}}\newtheorem{definition}{\textbf{Definition}}

\theoremstyle{definition}

\theoremstyle{remark}\newtheorem{remark}{\textbf{Remark}}

\begin{abstract}

Compressed sensing (CS) is about recovering a structured  signal from its under-determined linear measurements. Starting from sparsity, recovery methods have steadily moved towards  more complex structures.  Emerging machine learning tools such as  generative functions that are based on neural networks are able to learn general complex structures from training data. This makes  them potentially powerful tools for designing CS algorithms. Consider a desired class of signals $\cal Q$,  ${\cal Q}\subset\mathds{R}^n$, and  a corresponding   generative function $g:{\cal U}^k\to\mathds{R}^n$,  ${\cal U}\subset \mathds{R}$, such that $\sup_{{\bf x}\in {\cal Q}}\min_{{\bf u}\in\Uc^k}{1\over \sqrt{n}}\|g({\bf u})-{\bf x}\|\leq \delta$.  A recovery method based on $g$ seeks $g({\bf u})$ with minimum measurement error. In this paper,  the performance of such a recovery method is studied, under both noisy and noiseless measurements. In the noiseless case, roughly speaking,  it is proven that, as $k$ and $n$ grow without bound and $\delta$ converges to zero, if the number of measurements ($m$) is larger than  the input dimension of the generative model ($k$), then  asymptotically, almost lossless recovery is possible. Furthermore,  the performance of an efficient  iterative algorithm  based on projected gradient descent is studied. In this case,  an auto-encoder is used to define and enforce the source structure at the projection step. The auto-encoder is defined by encoder and decoder (generative)  functions  $f:\mathds{R}^n\to{\cal U}^k$ and  $g:{\cal U}^k\to\mathds{R}^n$, respectively.  We theoretically prove that,  roughly, given $m>40k\log{1\over \delta}$ measurements, such an algorithm converges to the vicinity of the desired result, even in the presence of additive white Gaussian noise. Numerical results exploring  the effectiveness of the proposed method are presented. 

\end{abstract}

\section{Introduction}

\subsection{Problem statement}

Solving inverse problems is at the core of many data acquisition systems, such as magnetic resonance imaging (MRI)  and optical coherence tomography~\cite{Huang91Science}. In many of such systems, through proper quantization in time or space,  the measurement system can be modeled as a system of linear equations as follows. The unknown signal to be measured is a high-dimensional signal $\xv\in \Qc$, where $\Qc$ represents  a compact subset of $\mathds{R}^n$. The measured signal can be represented  as  $\yv=A\xv+\zv$. Here $A\in\mathds{R}^{m\times n}$, $\yv\in\mathds{R}^m$ and  $\zv\in\mathds{R}^m$ denote the sensing matrix, the measurement vector, and the measurement noise, respectively. Typically the main goal is  to design an efficient   algorithm that recovers $\xv$ from the measurements $\yv$.   In addition to computational complexity, the efficiency of such an algorithm is measured in terms of its required number of measurements, its reconstruction quality, and its robustness to noise.  While  classic recovery methods were designed assuming that   $m$ is larger than $n$, i.e., the number of unknown parameters, during the last decade, researchers have shown that, since signals of interest are typically highly-structured, efficient recovery is possible, even if $m\ll n$.

The main focus in compressed sensing (CS), \ie solving   the described ill-posed linear inverse problem,  has been on  structures, such as  sparsity. Many signals of interest are indeed sparse or approximately sparse in some transform domain, which makes sparsity a fundamental structure, both from a theoretical and from a practical perspective. However, most of such signals of interest, in addition to being sparse, follow other more complex structures   as well. Enabling recovery algorithms to take advantage of the full structure of a class of signals could  considerably   improve the   performance. This has motivated researchers in CS  to explore algorithms that go beyond simple models such as sparsity. 

Developing a CS recovery method involves two major steps: i) studying the desired  class of signals (e.g., natural images, or MRI images) and discovering the structures that are shared   among them, and  ii) devising an efficient algorithm that given $\yv=A\xv+\zv$, finds a signal that is consistent with the measurements  $\yv $ and also the discovered structures. For instance, the well-known iterative hard thresholding algorithm \cite{blumensath2009iterative} is an algorithm that is developed for the case where the discovered  structure   is sparsity.  

One  approach to address the described  first step  is to design a method that  automatically learns complex signal models from training data. In other words,  instead of requiring domain experts to closely study a class of signals, we build an  algorithm that  discovers complex source models from   training data. 
While designing  such learning mechanisms is in general very complicated, generative functions (GFs) defined by  trained neural networks (NNs) present a successful modern tool in this area. The well-known universal approximation theory (UAT) states that with proper weights, NNs can approximate any regular function with arbitrary precision  \cite{cybenko1989approximations,funahashi1989approximate,hornik1989multilayer,barron1994approximation}. This suggests that  trained NNs operating as GFs are potentially  capable of  capturing complex unknown  structures. 

In recent years,  availability of i) large training data-sets on one hand, and ii)  computational tools such as GPUs on the other hand, has led to considerable progress in training effective NNs with state-of-art performance. While initially such networks were mainly trained  to solve classification problems, soon researchers realized that there is no fundamental reason to restrict our attention to such problems. And indeed researchers have explored application of NNs in  a wide range of applications including designing effective GFs.
%
The role of a GF is to learn the distribution of a class of signals, such that it is able to generate samples from that class. (Refer to Chapters 4 and 12  in \cite{friedman2001elements} to learn more about  using  GFs in classification.) Modern GFs achieve this goal typically through employing trained neural networks.  Variational auto-encoders (VAEs)  \cite{diederik2013variational} and generative adversarial nets (GANs) \cite{Goodfellow-et-al-2016} are examples of methods used to train complex GFs.  The success of such approaches in solving machine learning problems heavily relies on their ability to learn distributions of various complex signals, such as image and audio files. This success has encouraged researchers from other areas, such as compression, denoising and CS, to  look into the application of such methods, as  tools to capture the structures  of signals of interest. 

Given a class of signals, $\Qc\subset\mathds{R}^n$,  consider a corresponding  trained GF  $g:\Uc^k\to\mathds{R}^n$, $\Uc\subset\mathds{R}$. Assume that $g$ is trained by enough samples from $\Qc$, such that it is able to  represent signals from $\Qc$, possibly with some bounded loss. In this paper, we study the performance of an optimization-based CS recovery method that employs $g$ as a mechanism to capture the structure of signals in $\Qc$. We derive sharp bounds connecting the properties of function $g$ (its dimensions,  its error in representing the signals in $\Qc$, and its smoothness level) to the performance of the resulting recovery method.  We also study, both theoretically and empirically, the performance of an iterative  CS recovery method based on projected gradient descent (PGD) that employs $g$ to capture and enforce the source model (structure). We connect the number of measurements required by such a recovery method with the properties of function $g$.

\subsection{Notations}
Vectors are denoted by bold letters, such as $\xv$ and $\yv$. Sets are denoted by calligraphic letters, such as $\Ac$ and $\Bc$. For a set $\Ac$, $|\Ac|$ denotes its cardinality. For $x\in\mathds{R}$ and $b\in\mathds{N}^+$, $[x]_b$ denotes the $b$ bit quantized version of $x$ is defined as $[x]_b=2^{-b}\lfloor 2^bx\rfloor$. For a set $\Ac\subset\mathds{R}$ and $b\in\mathds{N}^+$, let $\Ac_b$ denote the set where every member in $\Ac$ is quantized in $b$ bits, \ie 
\[
\Ac_b\triangleq \{[x]_b: x\in\Ac\}.
\]

\subsection{Paper organization} 
Section \ref{sec:main} describes the problem of CS using GFs and states our main result on the performance of an optimization that employs a GF to capture the source structure. Section \ref{sec:algorithm} describes an efficient algorithm based on PGD to approximate the solution of the mentioned optimization which is based on exhaustive search. Section \ref{sec:related-work} reviews some related work in the literature.  Section \ref{sec:simulation} presents our simulation  results  on the performance of the  algorithm based on PGD. 
Section \ref{sec:proofs} presents the proofs of the main results and Section \ref{sec:conclusion} concludes the paper.


\section{Recovery using GFs}\label{sec:main}
Consider a class of signals represented  by a compact set $\Qc\subset\mathds{R}^n$. (For example, $\Qc$ can be the set of images of human faces, or the set of MRI images of human brains.)  Let function $g:\Uc^k\to\mathds{R}^n$ denote a GF trained to represent signals in set $\Qc$. (Throughout the paper, we assume that $\Uc$ is a bounded subset of $\mathds{R}$.) 
 \begin{definition}
 Function $g:\Uc^k\to\mathds{R}^n$ is said to cover set $\Qc$ with  distortion $\delta$, if 
\begin{align}
\sup_{\xv\in\Qc}\min_{\uv\in\Uc^k}{1\over \sqrt{n}}\|g(\uv)-\xv\| \leq \delta. \label{eq:def-delta-x}
\end{align}
\end{definition}
In other words, when function $g$  covers set $\Qc$ with  distortion $\delta$, it is able to represent all signals in $\Qc$ with a mean squared error less than $\delta^2$.

Consider the standard problem of CS, where instead of explicitly  knowing the structure of  signals in $\Qc$, we have access to function $g$, which is known to well-represent  signals in $\Qc$. In this setup, signal $\xv\in\Qc$ is measured as $\yv=A\xv+\zv$, where $A\in\mathds{R}^{m\times n}$, $\yv\in\mathds{R}^m$  and $\zv\in\mathds{R}^m$ denote the sensing matrix, the measurement vector, and the measurement noise, respectively. The goal is to recover $\xv$ from $\yv$, typically with  $m\ll n$, via  using the function $g$ to define the structure of signals in $\Qc$. 



To solve this problem,  ideally, we need to find a signal that is i) compatible with the measurements $\yv$, and ii)  representable with function $g$. Hence, ignoring the computational complexity issues, we would like to solve the following optimization problem:
\begin{align}
{\uvh}=\argmin_{\uv\in\Uc^k}\|Ag(\uv)-\yv\|,\label{eq:recovery-g-exhaustive-cont}
\end{align}
After finding $\uvh$, signal $\xv$ can be estimated as
\begin{align}
\xvh=g(\uvh).\label{eq:def-xvh-c}
\end{align}
The main goal of this  section is to theoretically study the performance of this  optimization-based recovery method. We derive bounds that establish a connection between the ambient dimension of the signal $n$, the parameters of the function $g$,  and the number of measurements $m$. 



To prove such theoretical results, we put some constraints on function $g$. More precisely, consider $\xv\in\Qc$ and let  and $\yv=A\xv+\zv$, where $A\in\mathds{R}^{m\times n}$ and $\zv\in\mathds{R}^m$.
Assume that 
\begin{enumerate}
\item  $g$ covers  $\Qc$ with distortion $\delta$, where $\delta\in(0,1)$, 
\item  $g$ is $L$-Lipschitz,
\item $\Uc$ is a bounded subset of $\mathds{R}$. 
\end{enumerate}   
  Define $\uvh$ and $\xvh$ as in \eqref{eq:recovery-g-exhaustive-cont} and \eqref{eq:def-xvh-c}, respectively. 
The following theorem characterizes the connection between the properties of function $g$ (input dimension $m$ and Lipschitz constant  $L$), the number of measurements ($m$) and the reconstruction distortion ($\|\hat{\xv}-\xv\|$). 

\begin{theorem}\label{lemma:1-noisy-main}
Consider compact set $\Qc\subset\mathds{R}^n$ and GF $g:\Uc^k\to\mathds{R}^n$ that covers $\Qc$ with distortion $\delta$. (Here, $\Uc$ is a compact subset of $\mathds{R}$.)
Consider $\xv\in\Qc$ and let $\yv=A\xv+\zv$, where $A\in\mathds{R}^{m\times n}$ and $\zv\in\mathds{R}^m$.  Assume that the entries of $A$ and $\zv$ are i.i.d.~$\Nc(0,{1\over n})$ and i.i.d.~$\Nc(0,\sigma^2)$, respectively. Define   $\uvh$ and $\xvh$  as \eqref{eq:recovery-g-exhaustive-cont} and \eqref{eq:def-xvh-c}.  Set free parameters $\eta>2$ and $\nu\in(0,1)$, such that ${1\over 2}-{\upsilon\over 2}-{1\over \eta}>0$.  Assume that $m\leq n$, and 
\begin{align}
m\geq \eta k.
\end{align} 
 Then,
\begin{align}
{1\over \sqrt{n}}\|\hat{\xv}-\xv\|  \leq &  \sqrt{6L\sigma} ({2k\over m})^{1\over 4}\delta^{{1\over 2}-{\upsilon\over 2}-{1\over \eta}}+4\sigma \delta^{-{2\over \eta}}\sqrt{k\ln {1\over \delta}\over m} +\alpha,
\end{align}
where $\alpha\triangleq    2\delta^{1-{1\over \eta}}+\delta^{{1\over 2}-{1\over \eta}}\sqrt{2\sigma} 
+ 3L\delta^{1-\upsilon -{1\over \eta}}\sqrt{k\over m}  +L\delta^{1-\upsilon}\sqrt{k\over n}=o(\delta^{{1\over 2}-{\upsilon\over 2}-{1\over \eta}})$,  with a probability larger than
\begin{align}
1- {\rm e}^{-(\upsilon - \zeta )k \ln{1\over \delta}}-  {\rm e}^{-k\ln {1\over \delta}}-3{\rm e}^{-0.8m},
\end{align}
where $\zeta=O({1\over \ln{1\over \delta}})$.
\end{theorem} 
The proof of Theorem \ref{lemma:1-noisy-main} is presented in Section \ref{sec:proof1-noisy}.

To better understand the implications of Theorem \ref{lemma:1-noisy-main},  the following corollary  considers the case of noiseless measurements (i.e.~$\sigma=0$). 

\begin{corollary}\label{cor:noiseless}
Consider the same setup as Theorem \ref{lemma:1-noisy-main}, where $\sigma=0$, i.e., the measurements are noise-free. Set free parameters $\eta>1$ and $\nu\in(0,1)$, such that $1-{\upsilon }-{1\over \eta}>0$.   If $m\geq \eta k$,  then with a probability larger than $1- {\rm e}^{-(\upsilon - \zeta )k \ln{1\over \delta}}-{\rm e}^{-0.8m}$,
\begin{align}
{1\over \sqrt{n}}\|\hat{\xv}-\xv\|  \leq &  {3L\over \sqrt{\eta}}  \delta^{1-{\upsilon }-{1\over \eta}}   +\alpha,
\end{align}
where $\alpha=o( \delta^{1-{\upsilon }-{1\over \eta}})$.
\end{corollary}
\begin{proof}
The proof is a straightforward application  of the proof of Theorem \ref{lemma:1-noisy-main}. Note that since there is no measurement noise in this case, we will not get error terms that are $O(\delta^{{1\over 2}-{\upsilon\over 2}-{1\over \eta }})$. Therefore, The condition on $\eta$ and $\nu$ here has changed to $\eta>1$ and $1-{\upsilon }-{1\over \eta}>0$.
\end{proof}

\begin{remark}
Consider a lossless  GF for a given class of signals described by $\Qc$, a compact subset of $\mathds{R}^n$. That is, $\sup_{\xv\in\Qc}\min_{\uv\in\Uc^k}\|\xv-g(\uv)\|=0$. In this case, $\delta=0$. In such a scenario, Corollary \ref{cor:noiseless} states that, essentially, $m>k$ measurements are sufficient for almost lossless recovery. 
\end{remark}

The optimization described in \eqref{eq:recovery-g-exhaustive-cont} was first proposed and analyzed in \cite{bora2017compressed}. It was shown in \cite{bora2017compressed} that $O(k\log L)$ measurements are sufficient for accurate recovery. However, in our results (Theorem \ref{lemma:1-noisy-main} and Corollary \ref{cor:noiseless}), the number of measurements does not scale with $L$ (Lipschitz constant) or $\delta$ and instead is proportional to $k$. This is consistent with our expectations  as the Minkowski dimension of a  class of signals that are generated by  a GF with input dimension  $k$ is also $k$, and therefore, in the noiseless setting, we expect to  be able to recover the signal from $k$ noise-free measurements \cite{WuVe10}.

\begin{remark}
In the presence of Gaussian noise, first note that, unlike prior work, the error terms in Theorem \ref{lemma:1-noisy-main} scale with the noise power ($\sigma$), rather than  $\|\zv\|$.   Moreover, the dominant error term that does not disappear as $\delta$ converges to zero is $4\sigma \delta^{-{2\over \eta}}\sqrt{k\ln {1\over \delta}\over m} $. To understand the role of this term, first note that, in the presence of Gaussian noise, due to the trade-off between bias and variance, it is not optimal to choose a model with $\delta$ close to zero. Instead, the optimal choice of $\delta$ would depend on the power of noise ($\sigma$), and as the noise power increases, models with larger values of $\delta$ will result in more accurate estimates. (Refer to Section \ref{sec:simulation} for numerical validation of this point.) Second, note that the mentioned error term scales with $m$ as $O({1\over\sqrt{m}})$. This implies that, for any noise power $\sigma$ and any representation error $\delta$, as the  number of measurements $m$ grows, the effect of this term vanishes as $O({1\over\sqrt{m}})$. 
\end{remark}

\section{AE-PGD algorithm}\label{sec:algorithm}

The optimization described in \eqref{eq:recovery-g-exhaustive-cont} is a challenging non-convex optimization. The GF $g$ defined using an NN   is a differentiable function. Therefore, one approach to solving  $\min_{\uv\in\Uc^k}\|Ag(\uv)-\yv\|$ is to   apply the  standard  gradient descent (GD) algorithm \cite{bora2017compressed}.  However,  since the problem is non-convex,  there is no guarantee that the solution derived based on this approach is close to the optimal solution. Another approach is to note that 
$\min_{\uv\in\Uc^k}\|Ag(\uv)-\yv\|\equiv \min_{\xvh\in\{g(\uv):\;\uv\in \Uc^k\}}\|A\xvh-\yv\|$ and   apply  PGD as follows: For $t=0,1,\ldots$, let
\begin{align}
\sv^{t+1}&=\xvh^t+\mu A^T(\yv-A\xvh^t)\nonumber\\
\uv^{t+1}&=\argmin_{\uv\in\Uc^k}\|\sv^{t+1}-g(\uv)\|\label{eq:uhat}\\
\xvh^{t+1}&=g(\uv^{t+1}).\label{eq:PGD-alg}
\end{align}
Still the described optimization is non-convex and therefore there is no guarantee that the algorithm will converge to the desired solution. The following theorem establishes this result and  connects the number of measurements $m$, the representation error of the GF $\delta$, the input dimension of the GF $k$, with the convergence performance of the PGD-based algorithm. Furthermore, it shows the robustness of this approach to additive white Gaussian noise.

\begin{theorem}\label{thm:4}
Consider $\xv\in\Qc$,   and $\yv=A\xv+\zv$, where  $\Qc$ denotes a compact subset of $\mathds{R}^n$ and $A\in\mathds{R}^{m\times n}$.  Here, $z_1,\ldots,z_m$ are i.i.d.~$\Nc(0,\sigma^2)$. Assume that function $g:[0,1]^k\to\mathds{R}^n$ is  $L$-Lipschitz and    satisfies  \eqref{eq:def-delta-x}, for some $\delta>0$.
Define  $\uvt$ and $\xvt$ as $\argmin_{\uv\in\Uc^k}\|\xv-g(\uv)\|$ and $\xvt=g(\uvt)$, respectively.   Choose free parameters $\alpha,\upsilon \in\mathds{R}^+$ and define $\eta$, $\gamma_1$ and $\gamma_2$ as 
\begin{align}
    \eta\triangleq  {k\over n} (1+ (\sqrt{n\over m}+2)^2) L^2\delta^{2\alpha},\label{eq:def-eta} 
\end{align}
\begin{align}
    \gamma_1\triangleq  (2+\sqrt{n\over m}\;)^2( L\delta^{\alpha}\sqrt{k\over n}+1)\label{eq:def-gamma1-thm2}
    \end{align}  
    and 
    \begin{align}
        \gamma_2\triangleq    { \sqrt{2k \over n}}(2+\sqrt{n\over m}\;),\label{eq:def-gamma2}
    \end{align}
respectively. 
Assume that
\begin{align}
{m \geq  40(1+\alpha+\upsilon)k\log {1\over \delta}.}
\end{align}  
Let $\mu={1\over m}$.  For $t=0,1,\ldots$, define $(\sv^{t+1},\uv^{t+1},\xvh^{t+1})$ as \eqref{eq:PGD-alg}. Then, for every $t$, if ${1\over \sqrt{n}}\|\xvh^{t}-\xvt\|\geq \delta$, then, either ${1\over \sqrt{n}}\|\xvh^{t+1}-\xvt\|\leq \delta$, or 
\[
{{1\over \sqrt{n}}\|\xvt-\xvh^{t+1}\|\;\leq\; {0.9+\eta \over \sqrt{n}} \|\xvt-\xvh^{t}\| +\left(\sqrt{6(1+\alpha)\left(\log{1\over \delta}\right)k\over m}+\gamma_2L\delta^{\alpha} \right){\sigma\over \sqrt{n}}+\gamma_1\delta,}
\]
with a probability larger than $1- 2^{-2k\upsilon \log{1\over \delta} }-\ex^{-{m\over 2}}- \ex^{-0.1 (1+\alpha)\left(\log{1\over \delta}\right)k+2(\ln 2)k}-\ex^{- 0.15m}$.
\end{theorem}
Theorem \ref{thm:4} states that although the original optimization is not  convex, having roughly more than $40k\log {1\over \delta}$ measurements, the described PGD  algorithm converges, even in the  presence of additive white Gaussian noise. 


In order to implement the proposed iterative method  described in \eqref{eq:PGD-alg}, the step that might seem challenging is the projection step, i.e., $\uv^{t+1}=\argmin_{\uv\in\Uc^k}\|\sv^{t+1}-g(\uv)\|$. Since the cost function is differentiable, one can use GD to solve it \cite{shah2018solving}. However, since the cost is not convex, there is no guarantee that the solution will be close to the optimal. Moreover, using GD to solve this optimization, adds to the computational complexity of the problem. Therefore, instead, we consider training  a separate neural network that approximates the solution to this optimization. Concatenating  this neural network with the  neural network that define $g$  essentially  yields an ``auto-encoder'' (AE)  that maps a high-dimensional signal into low-dimensions, and then back to its original dimension. Using this perspective, the last two steps of the algorithm, basically pass $\sv^{t+1}$ through an AE. (See Fig.~\ref{fig:PGD}.) We refer to the PGD-based algorithm where the projection is achieved by  an AE as the  AE-PGD method.

\begin{figure}[htbp!]
\begin{center}
\includegraphics[width=0.5\textwidth]{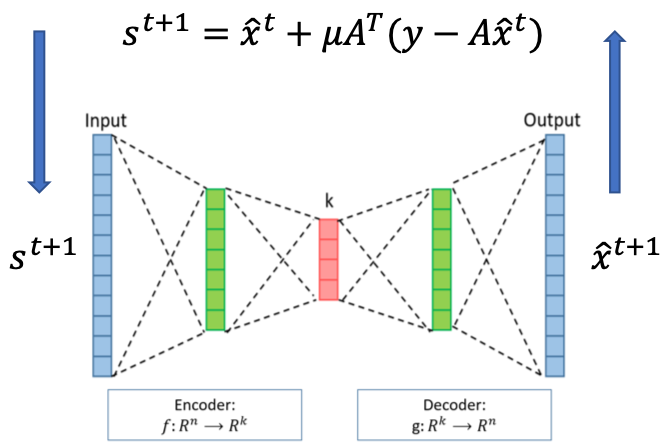}\caption{AE-PGD CS recovery. The top equation is the gradient descent and the bottom plot shows the AE employed to perform the projection of the signal.} \label{fig:PGD}
\end{center}
\end{figure}

 \section{Related work} \label{sec:related-work}
 Using NNs for CS has been an active area of research in recent years. (See \cite{mousavi2015deep,Kulkarni2016CVPR,rick2017one,borgerding2017amp,bora2017compressed,Yuan18OE,Jin0917TIP,shah2018solving,van2018compressed} for a non-comprehensive list of such results.) Closely studying the literature in this area reveals that, interestingly, the role imagined for the NN to play is not shared  among different approaches. While in some methods, NNs are directly trained to solve the inverse problem,  in others, they are trained,  independent of the CS recovery problem, as GFs that capture the source model. Our focus in this paper is on the latter type of methods where the role of the NN is to build a powerful GF that  captures the source complex structure. 
Application of NN-based GFs to solve CS  problems was first proposed in \cite{bora2017compressed}, which proved that roughly $O(k d \log n)$ measurements are enough for recovering the signal using the optimization described in \eqref{eq:recovery-g-exhaustive-cont}. ($d$ denotes the number of hidden layers.)
In \cite{shah2018solving}, an iterative algorithm based on PGD (similar to the one studied here) was proposed and studied. Here, we derive sharp theoretical guarantees for both  i) the exhaustive search method described in Section \ref{sec:main} and ii) the PGD-based algorithm. Our bounds directly connect the number of measurements with the properties of the GF, such as its input dimension and its representation quality. In both cases, we study the performance under additive white Gaussian noise as well. 

Another related line of work is on using  compression codes in designing efficient compression-based recovery methods \cite{jalali2016compression}. The goal of such methods it to  elevate the scope of structures used by CS algorithms to those used by compression codes. 
Such an optimization is  similar to ~\eqref{eq:def-xvh-c}. However, the difference between these two approaches is that while a lossy compression code can be represented by a \emph{discrete} set of codewords, a GF $g$ has a {\em continuous} input $\Uc^k$.

\section{Simulation results}\label{sec:simulation}
To further study the performance of the AE-PGD  recovery method, we examine  its performance on three different datasets:  i)  the MNIST hand-written digits  \cite{mnist}, ii) the chest X-ray images provided by NIH \cite{x-ray-images}, and iii) facial images from the CelebA dataset \cite{liu2015faceattributes}. 
The AE structure (2-layer encoder, and 2-layer decoder) (except the one reported in Section \ref{Sec:block-wise})  and the PGD algorithm are shown in Fig.~\ref{fig:PGD}. The implementations are performed in PyTorch using  Nvidia 1080 Ti GPU. We use the average peak signal-to-noise ratio  (PSNR)  to evaluate the quality of the reconstructed images. All the codes could be found at \cite{website:code2019}.

\begin{remark}
While Theorem \ref{thm:4} proves the convergence of the AE-PGD method for $\mu={1\over m}$, in our simulations,  we observed  that changing the step size could in fact improve  the performance. Therefore, using cross validation, in each setup, we  optimize the value of $\mu$. Potentially, one could further improve the performance by optimizing the step size at each iteration. However, this comes at a great computational complexity.  
\end{remark}


\subsection{MNIST \label{Sec:MNIST}}
In our first set of experiments,  we study the MNIST dataset. Each image in this dataset consists of $28\times 28$ pixels.  We use $35,000$ images for  training and $300$ images for testing. We consider an AE with fully-connected layers with sigmoid activation functions,  such that  the hidden layers of  the encoder and the decoder each  consists of $1,500$ hidden nodes.  We set the size of the output layer of the encoder and the input layer  of the GF ($k$) to $100$. The step size $\mu$ is set to $0.7$.

Fig.~\ref{fig:psnrmnist} compares the performance of the AE-PGD recovery with    Lasso \cite{tibshirani1996regression} and BM3D-AMP \cite{metzler2015bm3d} under different sampling  rates $m/n$, in both noise-free and noisy settings (middle plot corresponds to signal-to-noise-ratio (SNR) of $10$ dB).  It can be seen that in the {\em noise free} case,  when the sampling rate is low (e.g. $0.1$ and $0.05$), the AE-PGD method outperforms the other methods. When the sampling rate is higher (e.g. $0.2$ and $0.3$), BM3D-AMP achieves the best performance. 
In the {\em noisy} case, although BM3D-AMP still has the highest PSNR at high sample rates, its performance drops significantly. Some reconstructed images by the three algorithms (under noise free case) compared with the ground truth are shown in the right plot in Fig.~\ref{fig:psnrmnist}. 
\begin{figure}[!htbp]
	\includegraphics[width=0.35\textwidth]{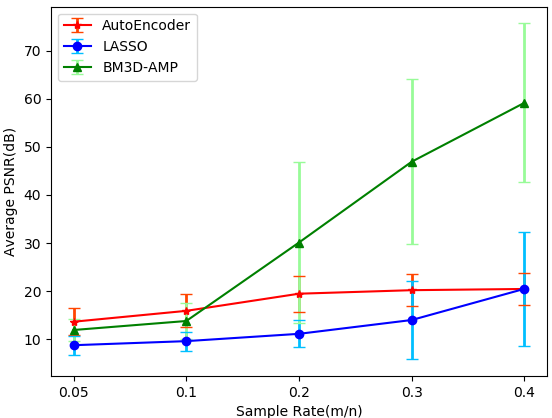}~
	\includegraphics[width=0.35\textwidth]{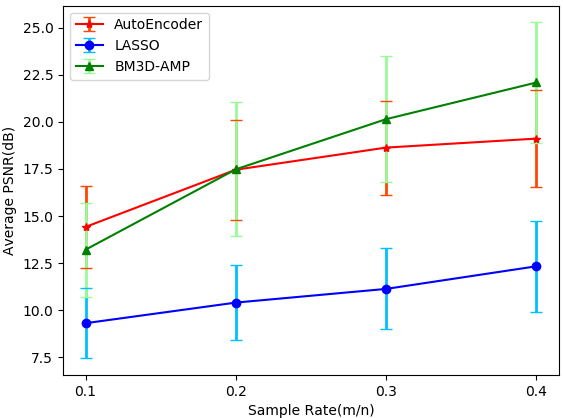}~
	\includegraphics[width=0.27\textwidth]{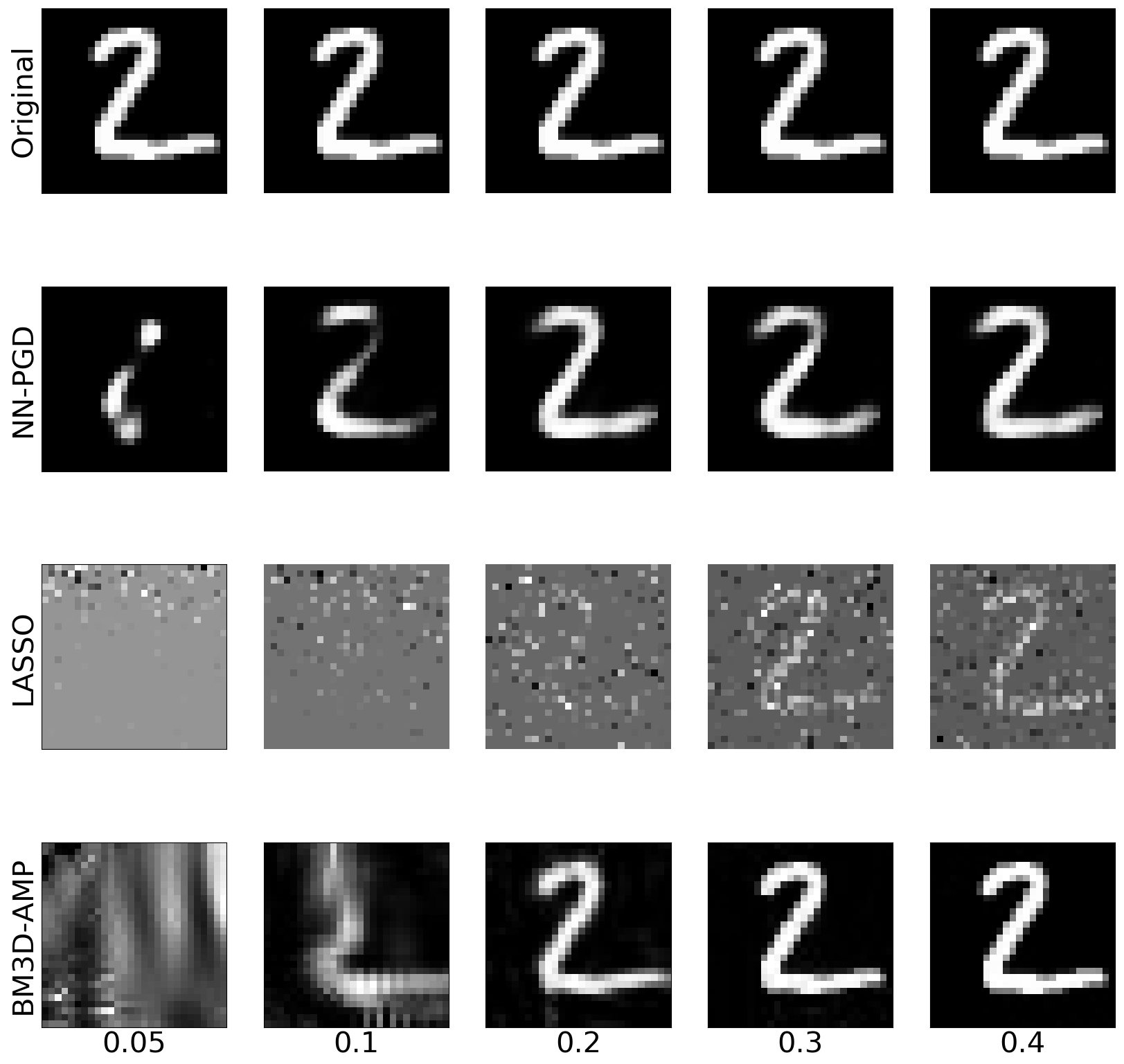}
	\caption{Comparing Lasso, BM3D-AMP and the proposed auto-encoder based inversion in the {\em noise free} case (left) and {\em noisy} case (middle with SNR = 10dB). Right: reconstructed images at different sampling rate in the noise free case.}
	\label{fig:psnrmnist}
\end{figure}

\subsection{X-ray Images \label{Sec:MNIST}}
We next explore the performance of the AE-PGD method on  chest X-ray images  \cite{x-ray-images}. 
In this dataset, each image  is of size $128\times 128$. We use $35,000$ training images and $100$ testing images. 
We compare the performance of the AE-PGD method with  BM3D-AMP and Lasso-DCT. This time, we consider two different  NNs calling the results NN1-PGD and NN2-PGD. In both cases, $\mu$ is set to $0.7$. Both NNs are structured as before with different number of nodes as follows:
\begin{itemize}
    \item [i)] NN1, $k=2000$ and there are 5000 hidden nodes in the first layer of the encoder and the second layer of the decoder; 
    \item [ii)] NN2, $k=3000$,  there are hidden 8000 hidden nodes in the first layer of the encoder and the second layer of the decoder.
\end{itemize}
In this case, all the activation functions, except those at the final layer of the decoder, are set to  rectified linear unit  (ReLU) function. The activation functions of the final layer are 
set as the sigmoid function. 

\begin{figure}[h]
	\centering
	\includegraphics[width=0.35\textwidth]{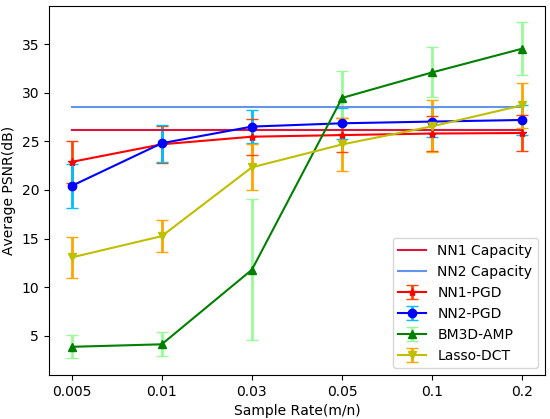}~
	\includegraphics[width=0.35\textwidth]{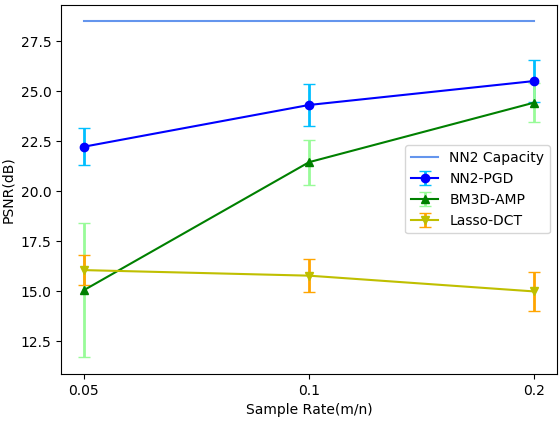}~
	\includegraphics[width=0.26\textwidth]{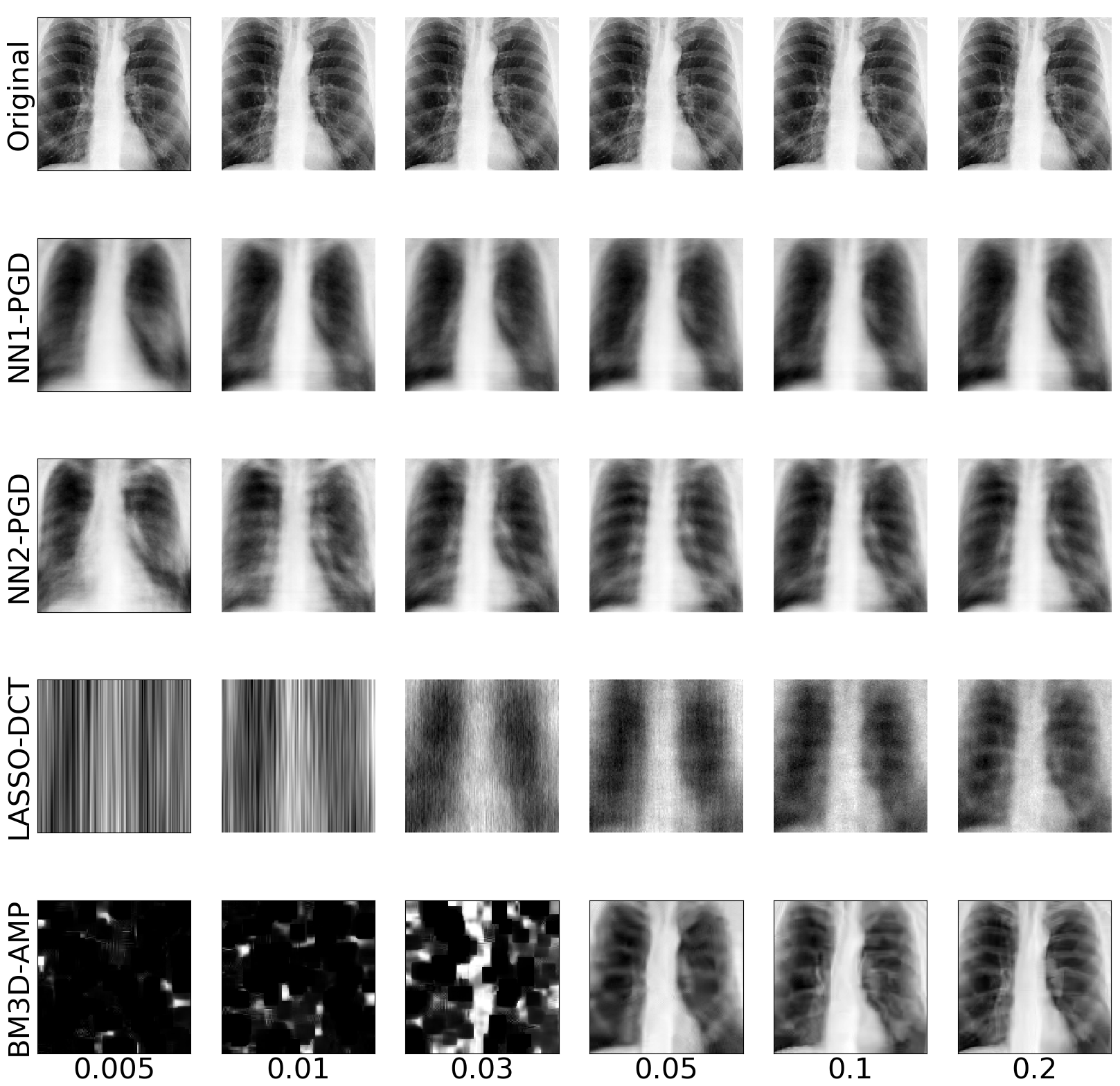}~
	\caption{PSNR of the reconstructed X-ray images under noise free (left) and noisy (middle with SNR = 10dB) cases and some example images (right).}
	\label{fig:psnr_xray}
\end{figure}

Fig.~\ref{fig:psnr_xray} shows the average PSNR on test images in both noiseless (left) and noisy (middle) settings, again at SNR = $10$ dB. The capacity of each NN refers to  the average representation error corresponding to each NN.  Clearly  the performance of the AE-PGD cannot exceed the capacity of the  NN it employs.  It can be observed that for both NNs, the AE-PGD method in fact achieves the capacity. This implies that  to improve the  performance of the AE-PGD method in high SNR regimes, one needs to design a NN with higher capacity, i.e., lower representation error.

Recall that Theorem \ref{thm:4} proves that the AE-PGD method converges, given enough measurement samples. To better understand the convergence behavior of the algorithm,  Fig. \ref{fig:iter} shows the average number of iterations of  the NN1-PGD and NN2-PGD methods, as a function of sampling rate. The step size is fixed  to $0.7$ as before. It can be seen that in both cases typically no more than $20$ iterations are required.

\begin{figure}[!htbp]
\centering
    \includegraphics[width=0.48\textwidth]{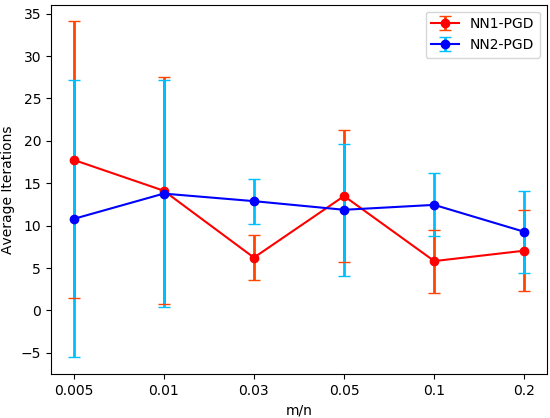}~
	\caption{Average number of iterations versus sampling rate for NN1-PGD and NN2-PGD}
	\label{fig:iter}
\end{figure}  
  

\subsection{Parallel Block-wise Neural Networks} \label{Sec:block-wise}
As shown in the previous section, the bottleneck in achieving high performance at higher sampling rates seem to be the accuracy of the representation error of the AE. In other words, to improve the performance of the AE-PGD method at higher sampling rates, we need to train AEs with representation error.   On the other hand,  the size of chest X-ray images suggests that to achieve this goal one  needs to train larger AEs.  Given our computational limitations, for example due to our GPU memory,  we next design a block-wise AE neural network, which breaks images into smaller blocks as follows. Again, the goal is to train a NN with higher capacity.  We crop  each image into four  smaller  $74\times 74$ images.  Then we train a separate AE consisting of  fours parts working  in parallel. Each part is an AE with the same structure as the one shown in Fig.~\ref{fig:PGD}, with  $k=3000$, and $8000$ hidden nodes for the other hidden layers.  
We allow some overlap between the image segments to avoid any blocking effects. For the pixels that are represented by more than one block, we take the average. As before, the step size is set to $0.7$.


\begin{figure}
    \centering
    \includegraphics[scale=0.6]{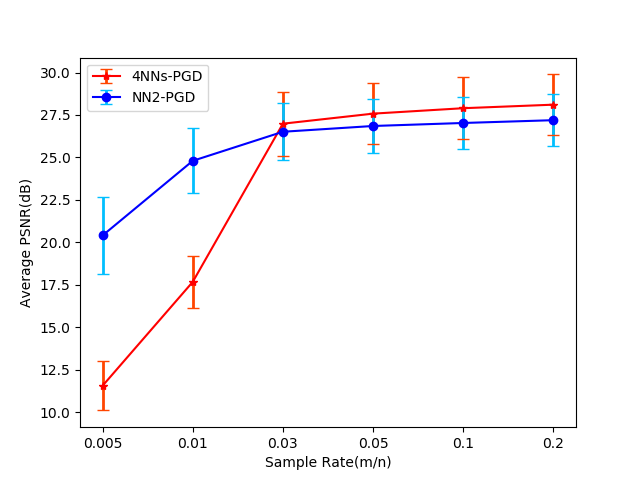}
    \caption{Average PSNR of block-wise NN (4NNS-PGD) and NN2-PGD.}
    \label{fig:psnr4NNs}
\end{figure}

In Fig. \ref{fig:psnr4NNs}, we compare the performance of the block-wise AE (referred to as 4NNs-PGD) with that of NN2-PGD (described in the previous section). It can be observed that, when the sampling rate is larger  than $0.03$, 4NNs-PGD achieves a better performance than NN2-PGD. On the other hand,   at lower sampling rates, the  network with a lower capacity (i.e., NN2) outperforms the 4NN network. We  also saw earlier that at lower sampling rates, the NN2 network outperforms BM3D-AMP. All these results show the power of AEs, as they can be designed to operate at different accuracies. In summary, the simulation results suggest that as the CS sampling rate grows, to achieve the best performance, one needs to adjust the accuracy  of the employed  AE accordingly. 

\subsection{U-net Refinement} \label{Sec:U-net}

 As shown in the previous section, training high-accuracy  AEs  is key to improving the performance of the  AE-PGD algorithm at higher sampling  rates.  Instead of directly improving the performance of the AE, in this section we explore a detour strategy as follows:  We train a U-Net \cite{U-net-ref} as a refinement function to improve the reconstructed image quality. 
To train the U-Net, we first train an AE (As described earlier) and then pass  the original training dataset  through the trained  AE to generate a new training dataset for U-Net. Then we use the original images and their reconstructions of the AE to train the U-Net such that the output is close to the original images. In other words, the U-NET receives the output of the AE and is expected to regenerate  the input of the AE, as much as possible.  After training the U-Net, we first use the PGD-AE method to find $\xvh$ as before, and then refine it by passing  it through the trained U-Net. (The U-Net structure is shown in Fig. \ref{fig:unet}.)

\begin{figure}[t]
    \centering
    \includegraphics[scale=0.6]{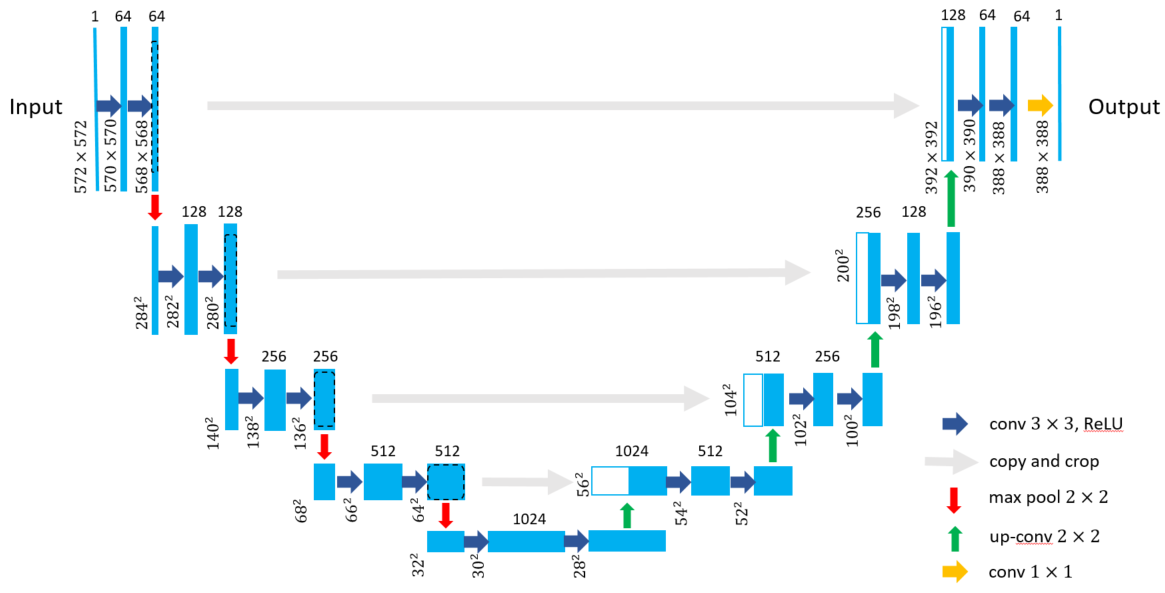}
    \caption{U-Net to refine the reconstruction results of the proposed PGD algorithm.}
    \label{fig:unet}
\end{figure}

\begin{figure}[t]
    \centering
    \includegraphics[scale=0.6]{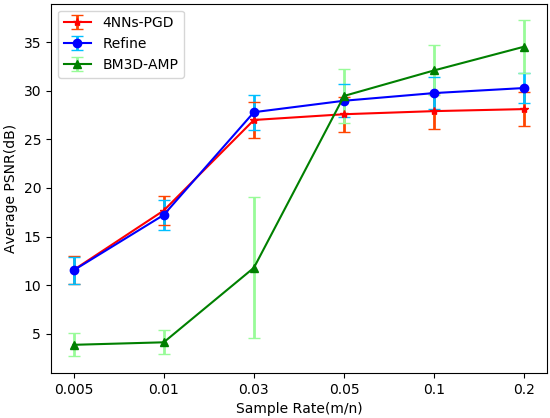}
    \caption{Average PSNR of the two stage NN (AE + U-Net).}
    \label{fig:psnrrefine}
\end{figure}

Fig. \ref{fig:psnrrefine} compares the performance of the AE-PGD with and without  U-NET refinement. Here, the AE is the blocked AE referred to as 4NN earlier. It can be observed  that this refinement step  improves the performance of the 4NNs-PGD method significantly at higher sampling rates, e.g., almost  $3$ dB at sampling rate $0.2$. However, the achieved performance  is still below that of BM3D-AMP when sampling  rate is high. 

It is worth noting that since the images in this dataset are rather noisy, and the U-Net seems to  perform  some denoising of the original images. On the other hand, the PSNR is calculated by comparing a reconstructed image with the original one.  Therefore, PSNR might not be an optimal measure to compare the performance of the algorithms. 
Inspecting the recovered images shown in Fig.~\ref{fig:image_rec_all} reveals that at  high sampling rates, BM3D-AMP  reconstructs  images that are very close to the original ones, by even recovering the noise. But the AE-PGD method with refinement reconstructs less-noisy images, which arguably include almost all the details of the original images.

\begin{figure}
     \centering
     \begin{subfigure}[b]{0.31\textwidth}
         \centering
     \includegraphics[width=1\textwidth]{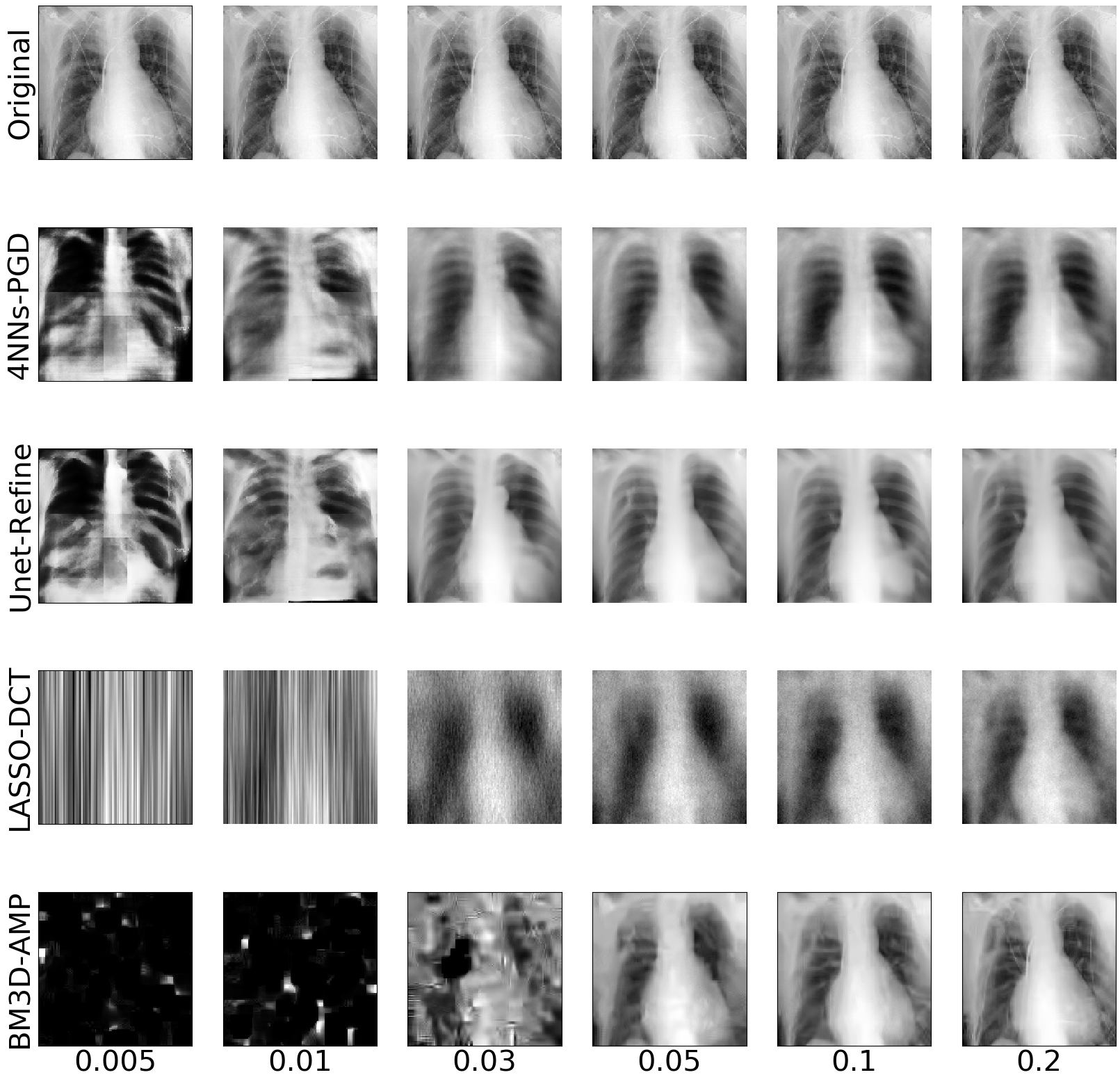}
     \end{subfigure}
     \hfill
     \begin{subfigure}[b]{0.31\textwidth}
         \centering
  \includegraphics[width=1\textwidth]{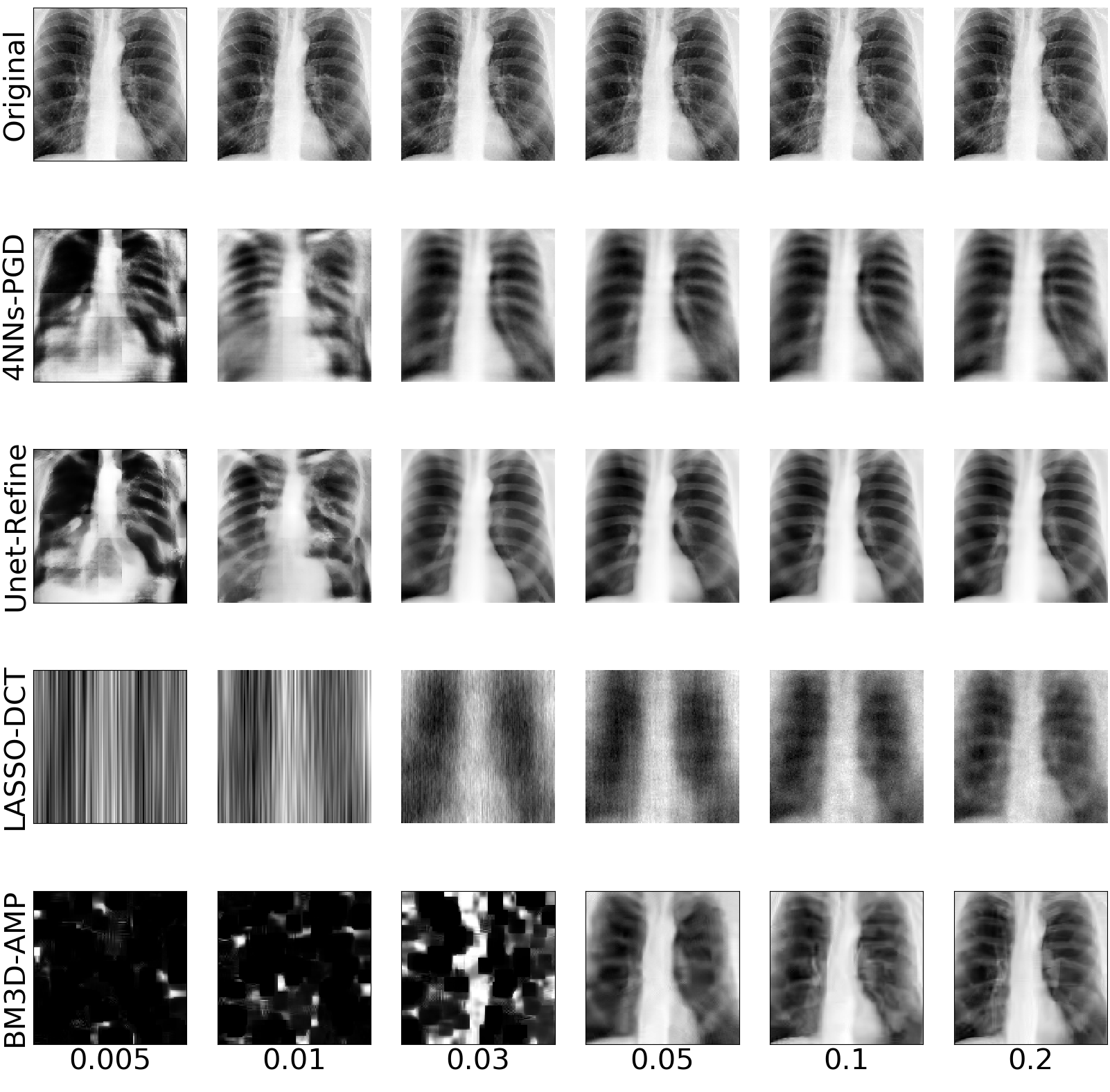}
     \end{subfigure}
     \hfill
     \begin{subfigure}[b]{0.31\textwidth}
         \centering
    \includegraphics[width=1\textwidth]{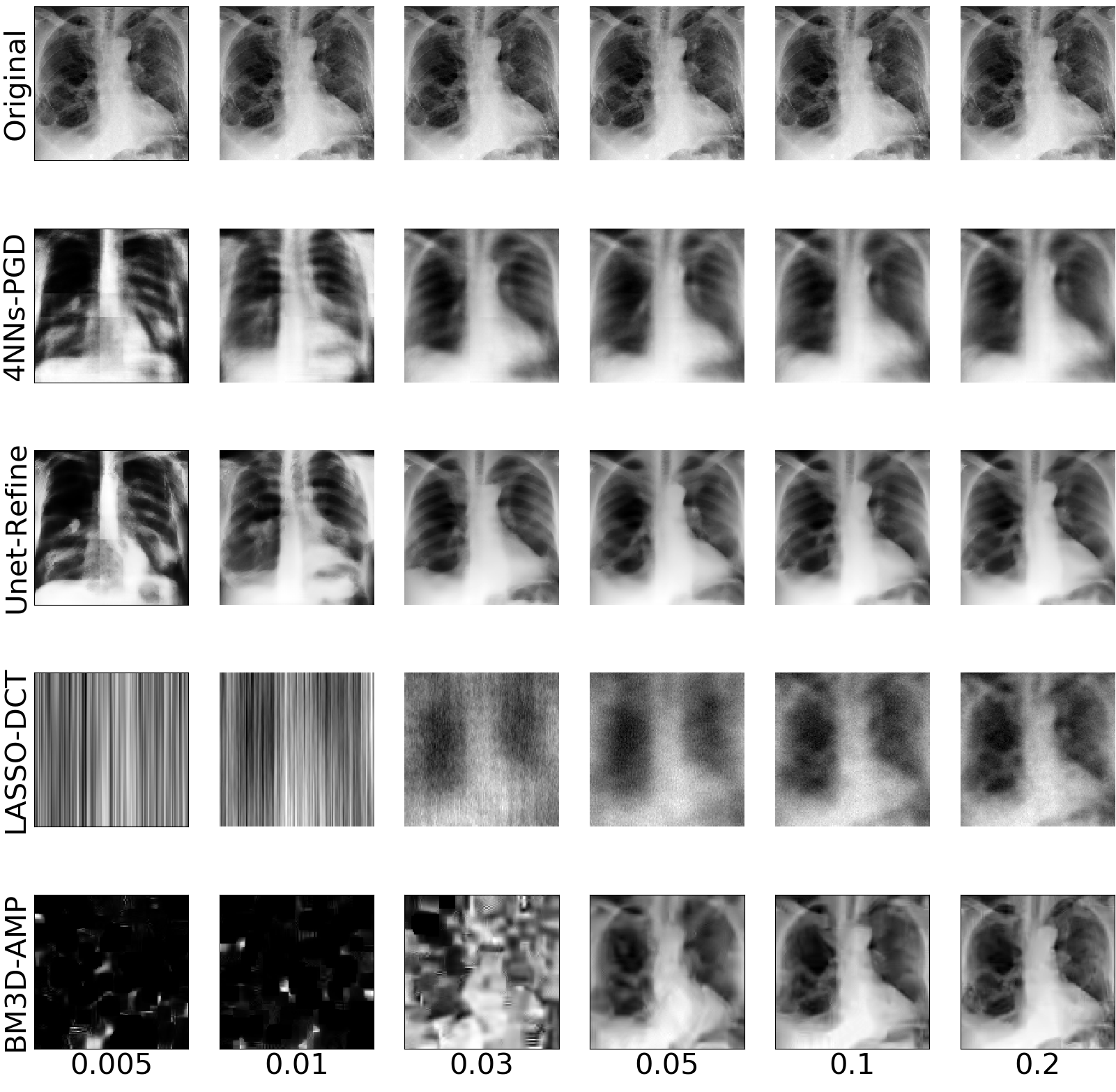}
     \end{subfigure}
 \caption{Exemplar reconstructed X-ray images by different algorithms compared with the original noisy images.}
    \label{fig:image_rec_all}
    \end{figure}


\subsection{Facial Images}
The small size of the digit images and the noisy nature of the X-ray images potentially pose as some limitations into the performance of any recovery method. As the final example, we test the AE-PGD  method on some clean facial images from the CelebA dataset \cite{liu2015faceattributes}. This time, each image consists of three $64\times64$ frames. We use  $50,000$ images for training, and $100$ images for test. We compare the performance of the NN-PGD method (with and without U-Net refinement)  with that of BM3D-AMP. For the AE, the input and output dimensions are set to $3\times 64\times$64 and $k=3000$. The number of hidden nodes in the encoder and the decoder are set to  $12,000$. All nodes in the AE are set to use the sigmoid activation function. As before, the U-Net is trained by using the images that are passed through the AE.

Fig. \ref{fig:psnrclean} show the performance of i) the AE-PGD method without refinement, and ii)  the BM3D-AMP method. For the AE-PGD, the step size $\mu$ is set to $(0.2,0.5,0.7,0.9)$ at sampling rates $(0.01,0.05,0.1,0.2)$, respectively. As before, the BM3D-AMP outperforms the AE-PGD method at higher sampling rates. The figure does not show the performance of the U-Net refinement as it has negligible effect in terms of PSNR. However,  while the refinement algorithm does not improve the performance much in terms of PSNR, as shown in Fig.~\ref{fig:image_clean_all},  it makes a  considerable visual impact on the quality of the recovered images.  

Carefully inspecting the figures shown in Fig.~\ref{fig:image_clean_all} simultaneously reveals some the   strengths and some of  the weaknesses of the AE-PGD method. The AE is essentially trained on human faces and therefore is capable of representing figures. On the other hand, its ability to capture the other details such as the background or accessories is limited. Therefore, comparing the images recovered by the AE-PGD method with those recovered by the BM3D-AMP reveals that while the former ones have better visual  qualities in terms of the faces themselves, still the overall PSNR of the latter group is better as they have a more uniform performance across the whole figure. 
\begin{figure}[t]
    \centering
    \includegraphics[scale=0.6]{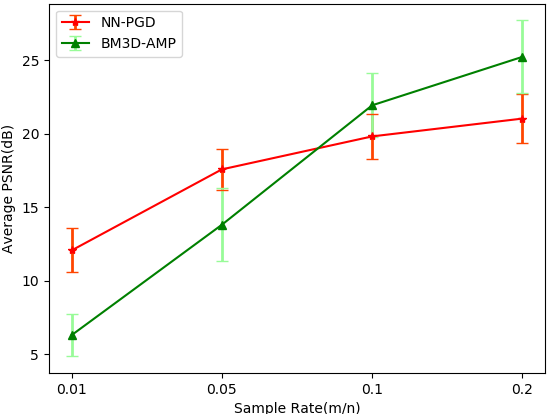}
    \caption{Average PSNR of 100 testing images in the CelebA dataset using different algorithms.}
    \label{fig:psnrclean}
\end{figure}

\begin{figure}
     \centering
     \begin{subfigure}[b]{0.31\textwidth}
         \centering
     \includegraphics[width=1\textwidth]{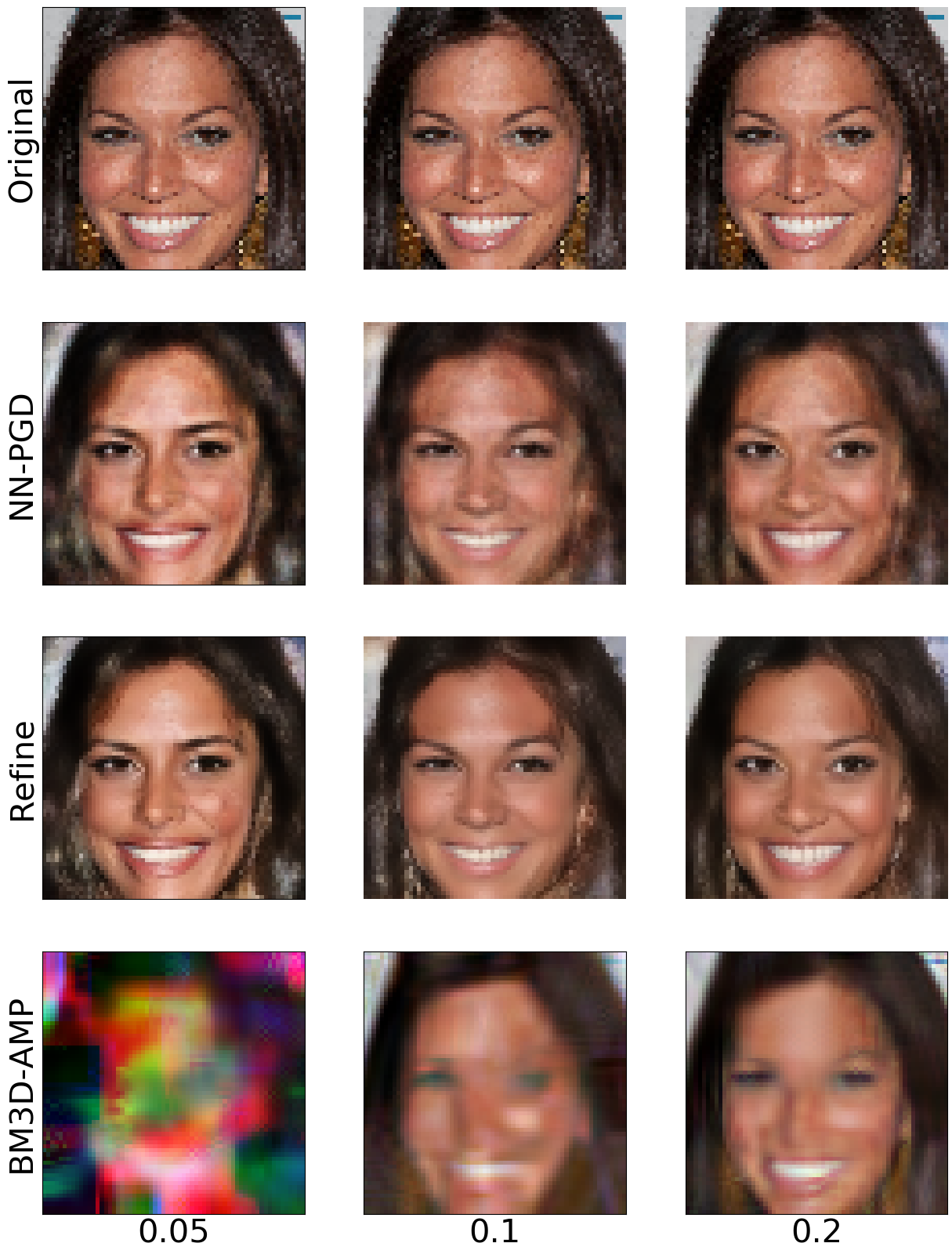}
     \end{subfigure}
     \hfill
     \begin{subfigure}[b]{0.31\textwidth}
         \centering
  \includegraphics[width=1\textwidth]{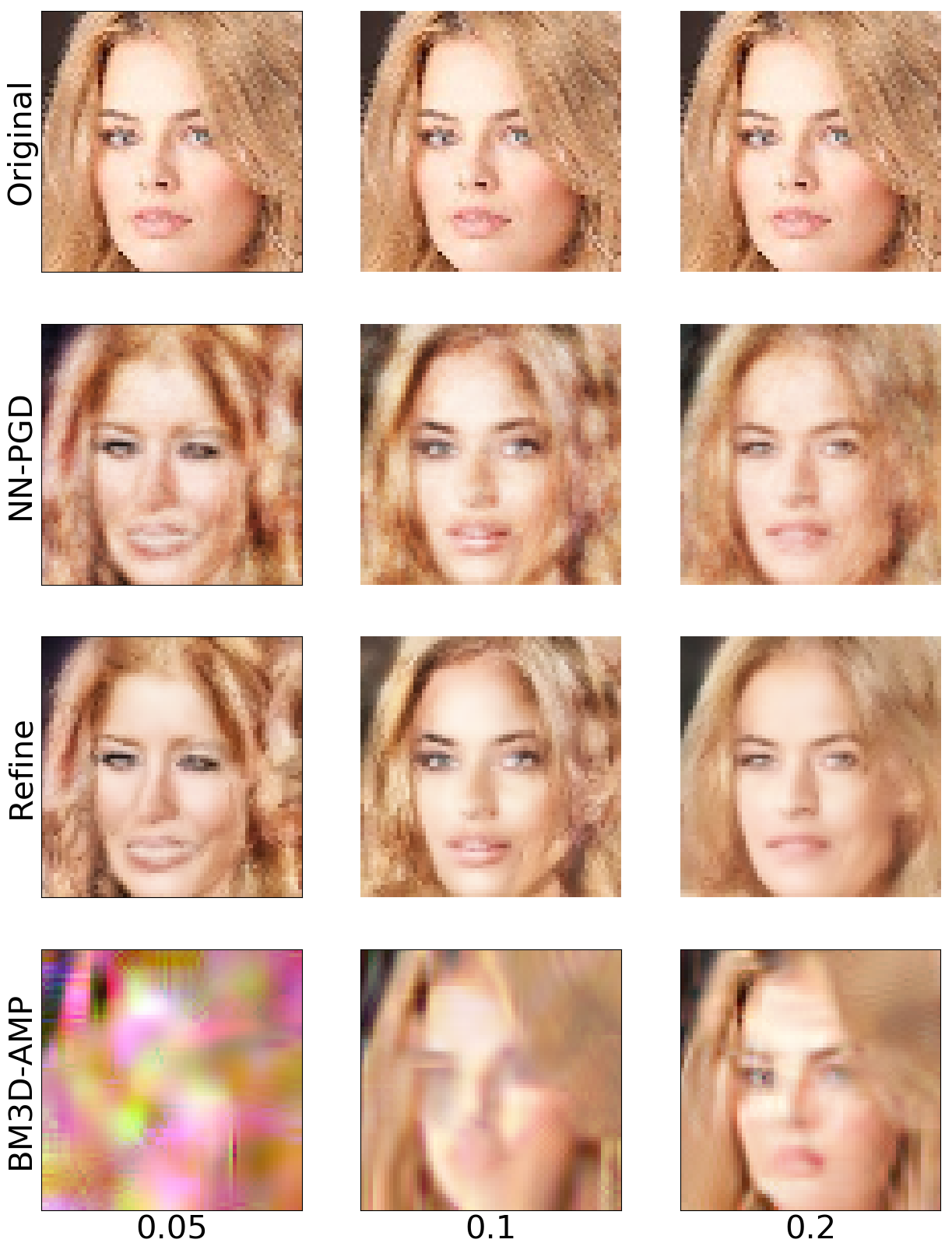}
     \end{subfigure}
     \hfill
     \begin{subfigure}[b]{0.31\textwidth}
         \centering
    \includegraphics[width=1\textwidth]{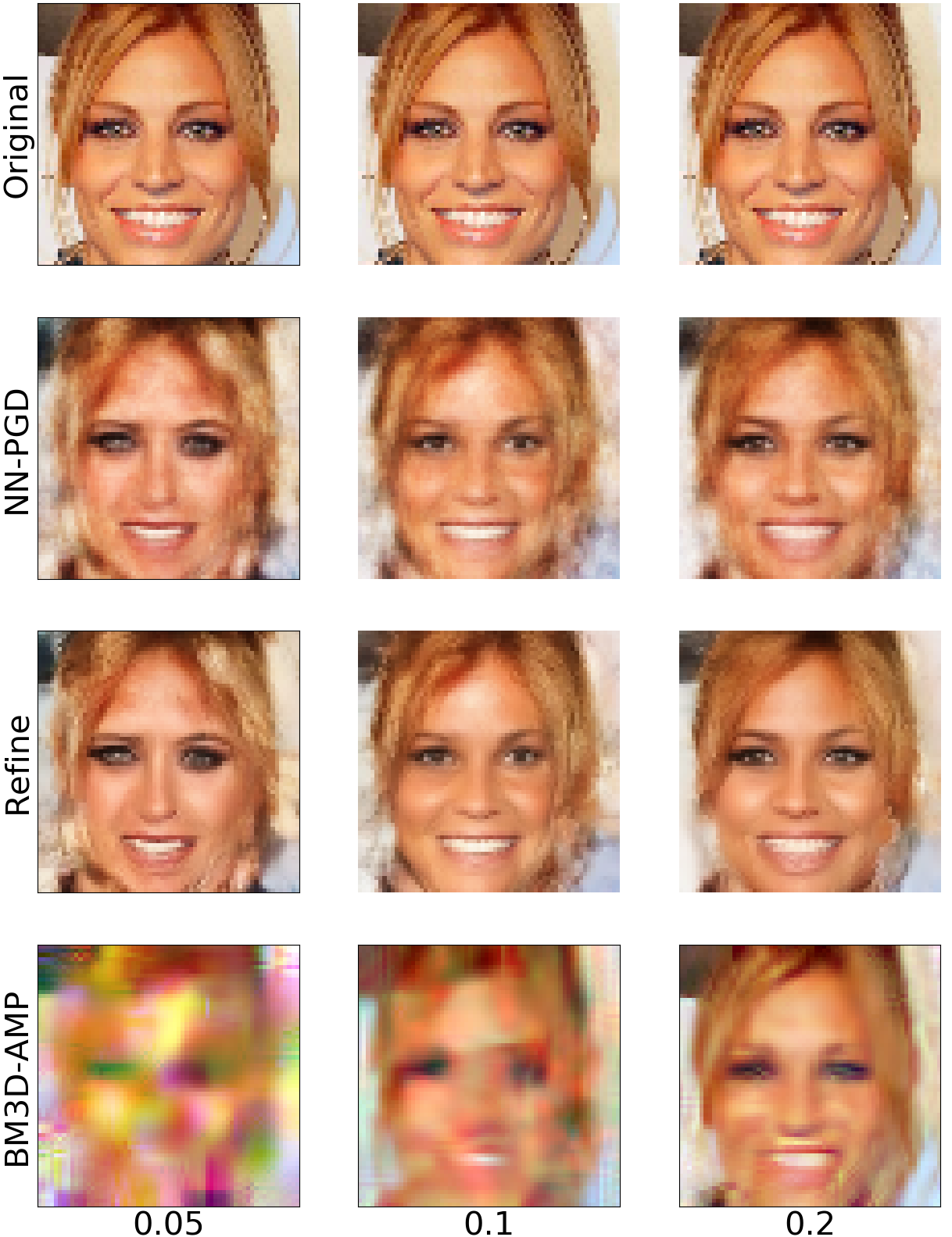}
     \end{subfigure}
        \caption{Reconstructed facial images by different algorithms at various sampling rates.}  \label{fig:image_clean_all}
\end{figure}

%

\section{Proofs}\label{sec:proofs}

The following lemma  from \cite{JalaliM:14-MEP-IT} on the concentration of Chi-squared random variables is used in the proof. 

\begin{lemma}[Chi-squared concentration] \label{chi-squared}
Assume that $U_1,\ldots,U_n$ are i.i.d.~$\Nc(0,1)$. For any $\tau \geq 0$ we have
\begin{equation}
 {\P\left(\sum_{i=1}^mU_i^2 > m (1 + \tau)\right) \leq {\rm e}^{- \frac{m}{2}(\tau - \ln(1 + \tau))}},
\end{equation}
 and for $\tau\in(0,1)$,
\begin{equation}
{\P(\sum_{i=1}^mU_i^2 <m (1 - \tau)) \leq {\rm e}^{ \frac{m}{2} (\tau + \ln(1 - \tau))}}.
\end{equation}
\end{lemma}

Also, the following lemma from \cite{jalali2018new} are used in the proof of Theorem \ref{thm:4}. 
\begin{lemma}\label{lemma:u-v-As-Av-concent}
Consider $\uv\in \mathds{R}^n$ and $\vv\in \mathds{R}^n$ such that $\|\uv\|=\|\vv\|=1$. Let $\a\triangleq \langle \uv,\vv \rangle $. Consider matrix $A\in\mathds{R}^{m\times n}$ with i.i.d.~standard normal entries. Then, for any $\tau>0$,
\begin{align}
\P\Big(\langle \uv,\vv\rangle-{1\over m}\langle A\uv,A\vv\rangle \geq \tau\Big)\leq \ex^{m((\a-\tau)s)-{m\over 2}\ln ((1+s\a)^2-s^2)},
\end{align}
where $s>0$ is a free parameter smaller than ${1\over 1-\a}$. Specifically, for $\tau=0.45$,
\begin{align}
\P\Big(\langle \uv,\vv\rangle-{1\over m}\langle A\uv,A\vv\rangle \geq 0.45\Big)\leq 2^{-0.05m}.
\end{align}
\end{lemma}

\begin{lemma}\label{lemma:gaussian-vectors}
Consider $\uv$ and $\vv$, where $u_1,\ldots,u_n, v_1,\ldots,v_n$ are  i.i.d.~$\Nc(0,1)$.   Then the distribution of  $\langle \uv,\vv \rangle=\sum_{i=1}^nu_iv_i$ is the same as the distribution of $\|\uv\|G$, where $G\sim\Nc(0,1)$ is independent of $\|\uv\|$.
\end{lemma}


\subsection{Proof of Theorem \ref{lemma:1-noisy-main}}\label{sec:proof1-noisy}

Define  $\tilde{\uv}$ and  $\tilde{\xv}$ as \eqref{eq:recovery-g-exhaustive-cont} and \eqref{eq:def-xvh-c}, respectively. 
Since $\hat{\uv}$ is the minimizer of $\|Ag(\uv)-\yv\|$, over all $\uv\in\Uc^k$, we have $\|Ag(\hat{\uv})-\yv\|\leq \|Ag(\tilde{\uv})-\yv\|.$ Moreover, by the triangle inequality,
\begin{align}
\|Ag([\hat{\uv}]_b)-\yv\| \leq \|Ag(\hat{\uv})-\yv\|+\|Ag(\uvh)-Ag([\hat{\uv}]_b)\|.
\end{align}
Recall that $\yv=A\xv+\zv$. Therefore, 
\begin{align}
\|A(g([\uvh]_b)-\xv)-\zv \| &\leq \|A(g(\tilde{\uv})-\xv)-\zv\|+\|A(g(\hat{\uv})-g([\hat{\uv}]_b))\|\nonumber\\
&\leq \|A(g(\tilde{\uv})-\xv)-\zv\|+\sigma_{\max}(A) L\|\hat{\uv}-[\hat{\uv}]_b\|\nonumber\\
&\leq \|A(g(\tilde{\uv})-\xv)-\zv\|+2^{-b}\sqrt{k}\sigma_{\max}(A) L.\label{eq:noisy-thm-step1}
\end{align}
Define $\ev_1$ and $\ev_2$  as  $\ev_1=g(\uvt)-\xv$ and $\ev_2=g([\uvh]_b)-\xv$, respectively. Define their normalizer versions as $\bar{\ev}_i=\ev_i/\|\ev_i\|$, $i=1,2$. Given $\tau_1>0$ and  $\tau_2\in(0,1)$, define events $\Ec_1$ and $\Ec_2$  as 
$\Ec_1\triangleq \{\|A\bar{\ev}_1 \|\leq \sqrt{{m\over n} (1+\tau_1)}\},$ and 
$\Ec_2\triangleq \Ec_2=\{\|A(g(\uv)-\xv)\|\geq \sqrt{{m\over n} (1-\tau_2)}\|g({\uv})-\xv\|:\;\forall \uv\in\Uc_b^k\}$, respectively. 
%
Furthermore, given $\tau_z>$, $\tau_3>0$ and $\tau_4>0$, define events $\Ec_a$, $\Ec_z$, $\Ec_3$ and $\Ec_4$ as
\begin{align}
\Ec_A\triangleq \Big\{\sigma_{\max}(A) \le 1+2\sqrt{m\over n}\Big\},
\end{align}
\begin{align}
\Ec_z\triangleq \{ \|\zv\|\leq \sigma\sqrt{m(1+\tau_z)}\;\},
\end{align}
\begin{align}
\Ec_3\triangleq \{|\langle A\bar{\ev}_1, \zv\rangle|\leq   \sigma \tau_3\sqrt{m\over n}\;\},
\end{align}
and
\begin{align}
\Ec_4\triangleq \{ |\langle A\bar{\ev}, \zv\rangle|\leq \sigma\tau_4 \sqrt{m\over n}: \bar{\ev}={g(\uv)-\xv\over \|g(\uv)-\xv\|}, \uv\in\Uc_b^k \},
\end{align}
respectively. From Lemma \ref{chi-squared}, 
$
\P(\Ec_1^c)\leq {\rm e}^{- \frac{m}{2}(\tau_1 - \ln(1 + \tau_1))},
$
and, for a fixed $\uv\in\Uc_b^n$, with a probability larger than ${\rm e}^{ \frac{m}{2}(\tau_2 + \ln(1 - \tau_2))}$, 
\begin{align}
\|A(g(\uv)-\xv)\|\geq m(1-\tau_2)\|g({\uv})-\xv\|.
\end{align}
Therefore, applying the union bound, it follows that
\[
\P(\Ec_2^c)\leq |\Uc_b|^k {\rm e}^{\frac{m}{2}(\tau_2 +\ln(1 - \tau_2))}.
\]
 Given a unit-norm vector $\bar{\ev}\in\mathds{R}^n$, $\sqrt{n}A\bar{\ev}$ is a random vector in $\mathds{R}^m$ with i.i.d.~$\Nc(0,1)$ entries. Therefore, according to Lemma \ref{lemma:gaussian-vectors}, $\sqrt{n} \langle A\bar{\ev}, \zv\rangle$ has the same distribution as $\|\zv\|G$, where $G\sim\Nc(0,1)$ is independent of $\|\zv\|$. 
On the other hand, using the law of total probability, $\P(\Ec_3\cap\Ec_4)\geq 1-  \P(\Ec_z^c)-\P(\Ec_3,\Ec_z)-\P(\Ec_4,\Ec_z)$, $i\in\{3,4\}$. Also 
\[
\P(\Ec_z^c)\leq {\rm e}^{-{m\over 2}(\tau_z-\ln(1+\tau_z))},
\]
and
\begin{align*}
\P(\Ec_3^c,\Ec_z)&= \P(|G| \|\zv\|>\tau_3\sigma \sqrt{m}, \|\zv\|\leq \sigma\sqrt{m(1+\tau_z)} )\nonumber\\
&\leq \P\Big(|G| >{\tau_3\over \sqrt{1+\tau_z}}\Big)\leq 2 {\rm e}^{-{\tau_3^2\over 1+\tau_z}}.
\end{align*}
Similarly, using the union bound, $\P(\Ec_4^c,\Ec_z)\leq 2|\Uc_b|^k {\rm e}^{-{\tau_4^2\over 1+\tau_z}}$.

Conditioned on $\Ec_A$, $2^{-b}\sqrt{k}\sigma_{\max}(A) L\leq  \Delta$,
where
\begin{align}
\Delta \triangleq  2^{-b}\sqrt{k} L\Big( 1+2\sqrt{m\over n}\Big),
\end{align}
Therefore, conditioned on $\Ec_A$, raising both sides of \eqref{eq:noisy-thm-step1} to power two and cancelling the common  $\|\zv\|^2$ term, it follows that
\begin{align}
\|A\ev_2\|^2-2\langle A\ev_2, \zv\rangle &\leq \|A\ev_1\|^2- 2\langle A\ev_1, \zv\rangle+2\Delta \|A\ev_1+\zv\|+\Delta^2\nonumber\\
&\leq \|A\ev_1\|^2- 2\langle A\ev_1, \zv\rangle+2\Delta (\|A\ev_1\|+\|\zv\|)+\Delta^2,
\end{align}
where the last line follows from the triangle inequality. Therefore, 
\begin{align}
\|\ev_2\|^2 \|A\bar{\ev}_2\|^2\leq\|\ev_1\|^2 \|A\bar{\ev}_1\|^2+2\|\ev_1\| |\langle A\bar{\ev}_1, \zv\rangle|+2\|\ev_2\| |\langle A\bar{\ev}_2, \zv\rangle|+2\Delta (\|\ev_1\|\|A\bar{\ev}_1\|+\|\zv\|)+\Delta^2.\label{eq:thm2-noisy-1}
\end{align}
Conditioned on $\Ec_1\cap\Ec_2\cap\Ec_3\cap\Ec_4$, noting that  ${1\over  \sqrt{n}}\|\ev_1\|\leq\delta$, it follows from \eqref{eq:thm2-noisy-1} that
\begin{align}
\|\ev_2\|^2 m(1-\tau_2) - 2\|\ev_2\| \sigma \tau_4 \sqrt{m n}- n(\gamma_1+\gamma_2)\leq 0,\label{eq:main-noisy-2}
\end{align}
where 
\begin{align}
\gamma_1 \triangleq (1+\tau_1) \delta^2m+ 2   \sigma \tau_3\delta\sqrt{m},  \label{eq:def-gamma1}
\end{align}
and
\begin{align}
\gamma_2\triangleq \Delta^2+2\Delta\sqrt{m}\Big(\sigma\sqrt{1+\tau_z}+\delta\sqrt{1+\tau_1}\Big). \label{eq:def-gamma2-thm2}
\end{align}
Since  \eqref{eq:main-noisy-2} is a second order equation, \eqref{eq:def-gamma2} implies that $\|\ev_2\|$ should be smaller than the largest root of this equation. Noting that  $\sqrt{a^2+b^2+c^2}\leq |a|+|b|+|c|$, for all $a$, $b$, and $c$ in $\mathds{R}$,  it follows that 
\begin{align}
{1\over \sqrt{n}}\|\hat{\xv}_b-\xv\|  \leq  {1\over  \sqrt{m} } \Big( {2\sigma \tau_4 \over 1-\tau_2}+\sqrt{\gamma_1\over 1-\tau_2}+\sqrt{\gamma_2\over 1-\tau_2}\;\Big).\label{eq:bound-e2-norm}
\end{align}

To finish the proof, we need to set the free parameters appropriately, such that the error probabilities converge to zero.  
\begin{itemize}
\item Set $b = \lceil(1-\upsilon)\log {1\over \delta}\rceil.$
\item Let  $\tau_1=3$. Therefore, $\P(\Ec_1^c)\leq {\rm e}^{- \frac{m}{2}(3 - \ln 4)}\leq {\rm e}^{- 0.8 m}$. 

\item Let $\tau_2=1-\delta^{{2\over \eta}}$. Since $\Uc$ is a compact set, there exist integer numbers  $a_1$ and $a_2$, such that $\Uc\subseteq [a_1,a_2]$. Therefore, $ |\Uc_b|\leq a2^{-b}$, where $a\triangleq a_2-a_1$.  Therefore, since, for $\tau_2\in(0,1)$,   $(\tau_2 + \ln(1 - \tau_2))\leq 0$ and $m\geq \eta k$ by assumption, it follows that 
\begin{align}
& {k \ln |\Uc_b|+\frac{m}{2}(\tau_2 + \ln(1 - \tau_2))}\nonumber\\
&{\leq k( \ln a+b\ln2 )+{\eta k\over 2} (\tau_2 + \ln(1 - \tau_2))}\nonumber\\
&{\leq k( \ln a-{(1-v)}\ln \delta )+{\eta k\over 2} (\tau_2 + \ln(1 - \tau_2))}, \label{eq:bound-expo}
\end{align}
where the last line follows because $b=\lceil (1-\upsilon) \log {1\over \delta }\rceil$ and hence $b\ln 2\leq  (1-\upsilon) \log{1\over \delta} \ln 2 =-(1-\upsilon)\ln \delta$. Therefore, from \eqref{eq:bound-expo}, inserting the value of $\tau_2$, we have 
\begin{align}
&k \ln |\Uc_b|+\frac{m}{2}(\tau_2 + \ln(1 - \tau_2))\nonumber\\
&{\leq k( \ln a-(1-\upsilon)\ln \delta )+{\eta k\over 2} {(1 {-\delta^{2\over\eta}} +{2\over \eta} \ln\delta)}} \nonumber\\
&{= - k (\upsilon - \zeta) \ln {1\over \delta}},
\end{align}
where
\vspace{-4mm}
\begin{align}
{\zeta={\ln a +{\eta \over 2} {(1-\delta^{2\over\eta})} \over \ln{1\over \delta}}.}\label{eq:def-zeta}
\end{align}
{Note that $\zeta$ only depend on $a$, $\eta$ and $\delta$ and $\zeta=O(1/\ln{1\over \delta})$. } Therefore, $\P(\Ec_2)\leq {\rm e}^{-(\upsilon - \zeta )k \ln{1\over \delta}}$.

\item Set $\tau_z=1$. Then, $\P(\Ec_z^c)\leq {\rm e}^{-{m\over 2}(1-\ln 2)}\leq {\rm e}^{-0.15 m}$.  
\item Set $\tau_3=\sqrt{m}$. As proved earlier, $\P(\Ec_3,\Ec_z^c)\leq  2 {\rm e}^{-{\tau_3^2\over 1+\tau_z}}$. Hence, for $\tau_3=\sqrt{m}$ and $\tau_z=1$, $\P(\Ec_3,\Ec_z^c)\leq {\rm e}^{-{m\over 2}}$. 
\item  Set $\tau_4=2\sqrt{k\ln {1\over \delta} }$. We need to  set $\tau_4$ such that $|\Uc_b|^k {\rm e}^{-{\tau_4^2\over 1+\tau_z}}$ converges to zero, as the dimensions of the problems grow. Note that 
$\ln(|\Uc_b|^k {\rm e}^{-{\tau_4^2\over 1+\tau_z}})=kb \ln 2-{1\over 2}\tau_4^2= k\lceil(1-\upsilon)\log {1\over \delta} \rceil \ln 2-{1\over 2}\tau_4^2\leq  k\ln {1\over \delta} -{1\over 2}\tau_4^2$. Setting $\tau_4=2\sqrt{k\ln {1\over \delta} }$, it follows that $\ln(|\Uc_b|^k {\rm e}^{-{\tau_4^2\over 1+\tau_z}})\leq -k\ln {1\over \delta}$. 
\end{itemize}
For the selected values of the parameters, we have
\begin{align}
\gamma_1 =2(2 \delta+    \sigma) \delta m, 
\end{align}
and 
\begin{align}
\gamma_2= \Delta^2+2\Delta\sqrt{m}(\sigma\sqrt{2}+2\delta).\label{eq:def-gamma2}
\end{align}
But, since  by assumption,  $m\leq n$, $\Delta \leq   3 2^{-b}\sqrt{k} L\leq 3 \sqrt{k} L\delta^{1-\upsilon}$. Therefore, 
\[
\gamma_2\leq  9k  L^2\delta^{2-2\upsilon}+6 \sqrt{km} L\delta^{1-\upsilon}(\sigma\sqrt{2}+2\delta).
\]
In summary, combining the bounds on $\gamma_1$ and $\gamma_2$ with \eqref{eq:bound-e2-norm}, it follows that
\begin{align}
{1\over \sqrt{n}}\|\hat{\xv}_b-\xv\|  \leq   4\sigma \delta^{-{2\over \eta}}\sqrt{k\ln {1\over \delta}\over m} +  \delta^{{1\over 2}-{1\over \eta}}\sqrt{2(2 \delta+    \sigma)}   +\delta^{{1\over 2}-{\upsilon\over 2}-{1\over \eta}} \sqrt{9({k\over m})  L^2\delta^{1-\upsilon}+6 \sqrt{k\over m} L(\sigma\sqrt{2}+2\delta)}.
\end{align}
Finally, to finish the proof, note that 
\begin{align}
\|\hat{\xv}-\xv\|\leq \|\hat{\xv}-\hat{\xv}_b\|+\|\hat{\xv}_b-\xv\|\leq L\|\uv-[\uv]_b\|+\|\hat{\xv}_b-\xv\|\leq L2^{-b}\sqrt{k}+\|\hat{\xv}_b-\xv\|.
\end{align}
Therefore, since $\sqrt{a^2+b^2+c^2}\leq |a|+|b|+|c|$, we have
\begin{align}
{1\over \sqrt{n}}\|\hat{\xv}-\xv\|  \leq &  
4\sigma \delta^{-{2\over \eta}}\sqrt{k\ln {1\over \delta}\over m} +  2\delta^{1-{1\over \eta}}+\delta^{{1\over 2}-{1\over \eta}}\sqrt{2\sigma}  \nonumber\\
& + 3L\delta^{1-\upsilon -{1\over \eta}}\sqrt{k\over m}  +\sqrt{6L\sigma} ({2k\over m})^{1\over 4}\delta^{{1\over 2}-{\upsilon\over 2}-{1\over \eta}} +2\sqrt{3} ({k\over m})^{1\over 4} \delta^{1-{\upsilon\over 2}-{1\over \eta}}+L\delta^{1-\upsilon}\sqrt{k\over n},
\end{align}
where concludes the proof as $\alpha\triangleq    2\delta^{1-{1\over \eta}}+\delta^{{1\over 2}-{1\over \eta}}\sqrt{2\sigma} 
 + 3L\delta^{1-\upsilon -{1\over \eta}}\sqrt{k\over m}  +L\delta^{1-\upsilon}\sqrt{k\over n}=o(\delta^{{1\over 2}-{\upsilon\over 2}-{1\over \eta}})$.

\subsection{Proof of Theorem \ref{thm:4}}\label{sec:proof4}

Recall that  ${\uvh}=\argmin_{\uv\in\Uc^k}\|g(\uv)-\xv\|$ and ${\xvh}=g({\uvh}).$ Since $\xvh^{t+1}=\argmin_{u^k\in\Uc^k}\|\sv^{t+1}-g(\uv)\|$, 
\[
\|\sv^{t+1}-\xvh^{t+1}\|\leq \|\sv^{t+1}-{\xvh}\|.
\]
But $\|\sv^{t+1}-\xvh^{t+1}\|^2=\|\sv^{t+1}-\xvh+\xvh-\xvh^{t+1}\|^2=\|\sv^{t+1}-\xvh\|^2+\|\xvh-\xvh^{t+1}\|^2+2\langle\sv^{t+1}-\xvh,\xvh-\xvh^{t+1} \rangle $. Therefore, 
\begin{align}
\|\xvh-\xvh^{t+1}\|^2\leq& 2\langle\xvh-\sv^{t+1},\xvh-\xvh^{t+1} \rangle\nonumber\\
= &2\langle\xvh-\xvh^{t},\xvh-\xvh^{t+1}\rangle-2 \mu\langle A({\xvh}-\xvh^{t}),A(\xvh-\xvh^{t+1})\rangle\nonumber\\
&-2 \mu\langle A(\xv-{\xvt}),A(\xvh-\xvh^{t+1})\rangle-\mu \langle A^T\zv, \xvh-\xvh^{t+1}\rangle.\label{eq:step-1}
\end{align}
For $t=1,2,\ldots$, define a normalized error vector as follows
\begin{align}
{\ev}^{t}={\xvh-\xvh^{t}\over \|\xvh-\xvh^{t}\|}.\label{eq:def-ev-t}
\end{align}
Using this definition, the triangle inequality and the Cauchy-Schwartz inequality, we rewrite \eqref{eq:step-1} as follows
\begin{align}
\|\xvh-\xvh^{t+1}\|\leq& 2\left(\langle{\ev}^{t+1},{\ev}^{t} \rangle-\mu \langle A{\ev}^{t+1},A{\ev}^{t} \rangle \right) \|\xvh-\xvh^{t}\|+2\mu\|A(\xv-{\xvh})\| \|A\ev^{t+1}\|\nonumber\\
&+2\mu \left|\langle A^T\zv, \ev^{t+1}\rangle\right|\nonumber\\
\leq& 2\left(\langle{\ev}^{t+1},{\ev}^{t} \rangle-\mu \langle A{\ev}^{t+1},A{\ev}^{t} \rangle \right) \|\xvh-\xvh^{t}\|+2\mu(\sigma_{\max}(A))^2\|\xv-{\xvh}\|\nonumber\\
&+2\mu \left|\langle A^T\zv, \ev^{t+1}\rangle\right|.\label{eq:step-2}
\end{align}
 To prove the desired result that connects the error at iteration $t+1$, $\|\xvh-\xvh^{t+1}\|$, to the error at iteration $t$, $ \|\xvh-\xvh^{t}\|$,  we first define the quantized versions of the error  and the reconstruction  vectors. The reason for this discretization becomes clear later when we use them to prove our concentration results.
 
For $t=1,2,\ldots$, define $\uv_b^{t}\triangleq [\uv^{t}]_b$ and 
\[
\xvh_b^{t}\triangleq g(\uv_b^{t}).
\] Also, let 
\[
\mvec{\eta}^t\triangleq \xvh^{t}-\xvh_b^{t}. 
\]
Assume that the quantization level $b$ is selected as follows
\begin{align}
b =  \left\lceil(1+\alpha)\log {1\over \delta}\right\rceil.\label{eq:def-b}
\end{align}

Since by assumption $g$ is a Lipschitz function, we have 
\begin{align}
\|\mvec{\eta}^t\|&=\|g(\uv^t)-g(\uv_b^t)\|\leq L\|\uv^t-\uv_b^t\|\leq L2^{-b}\sqrt{k}\nonumber\\
&\leq  L\delta^{1+\alpha}\sqrt{k}\label{eq:def-eta-k}
\end{align}
where the last line follow from \eqref{eq:def-b}.
Let 
\[
\ev_b^t\triangleq {\xvh-\xvh_b^{t}\over \|\xvh-\xvh_b^{t}\|}.
\]
We next bound $\|\ev^t-\ev^t_b\|$, the distance between  $\ev_b^t$ and $\ev^t$, where $\ev^t$ is defined in \eqref{eq:def-ev-t}. Note that 
\begin{align}
{\ev}^{t}&={\xvh-\xvh^{t}\over \|\xvh-\xvh^{t}\|}={\xvh-\xvh_b^{t}-\mvec{\eta}^t \over \|\xvh-\xvh_b^{t}-\mvec{\eta}^t\|}\nonumber\\
&=\ev_b^t-{\xvh-\xvh_b^{t}\over \|\xvh-\xvh_b^{t}\|}+{\xvh-\xvh_b^{t}-\mvec{\eta}^t \over \|\xvh-\xvh_b^{t}-\mvec{\eta}^t\|}.
\end{align}
Therefore, by the triangle inequality, it follows that 
\begin{align}
\|\ev^t-\ev_b^t\|& \leq {\left| \|\xvh-\xvh_b^{t}\|- \|\xvh-\xvh_b^{t}-\mvec{\eta}^t\|\right|\over  \|\xvh-\xvh_b^{t}-\mvec{\eta}^t\|}+{\|\mvec{\eta}^t\| \over \|\xvh-\xvh_b^{t}-\mvec{\eta}^t\|}\nonumber\\
&\leq {2\|\mvec{\eta}^t\| \over \|\xvh-\xvh_b^{t}-\mvec{\eta}^t\|}\nonumber\\
&= {2\|\mvec{\eta}^t\| \over \|\xvh-\xvh^{t}\|}\nonumber\\
&\stackrel{\rm (a)}{\leq} {2L2^{-b}\sqrt{k} \over \|\xvh-\xvh^{t}\|}\nonumber\\
&\stackrel{\rm (b)}{\leq}{ 2L\delta^{1+\alpha}\sqrt{k} \over \sqrt{n} \delta}=2L\delta^{\alpha}\sqrt{k\over n},\label{eq:dif-evt-evbt}
\end{align} 
where $\rm (a)$ and $\rm (b)$ follow from \eqref{eq:def-eta-k} and our assumption that $\|\xvh-\xvh^{t}\|\geq \sqrt{n}\delta$.

Using the introduced quantizations, in the following, we bound the three terms on the RHS of \eqref{eq:step-2}. 
\begin{itemize}
\item  $2\left(\langle{\ev}^{t+1},{\ev}^{t} \rangle-\mu \langle A{\ev}^{t+1},A{\ev}^{t} \rangle \right) \|\xvh-\xvh^{t}\|$: 
First, note that 
\begin{align}
\langle{\ev}^{t+1},{\ev}^{t} \rangle-\mu \langle A{\ev}^{t+1},A{\ev}^{t} \rangle &= \langle{\ev}^{t+1}-\ev^{t+1}_b+\ev^{t+1}_b,{\ev}^{t}-\ev^{t}_b+\ev^{t}_b \rangle\nonumber\\
&\;\;\;-\mu \langle A({\ev}^{t+1}-\ev^{t+1}_b+\ev^{t+1}_b),A({\ev}^{t}-\ev^{t}_b+\ev^{t}_b) \rangle\nonumber\\
&=\langle{\ev}_b^{t+1},{\ev}_b^{t} \rangle-\mu \langle A{\ev}_b^{t+1},A{\ev}_b^{t} \rangle\nonumber\\
&\;\;\;+\langle{\ev}^{t+1}-\ev^{t+1}_b,{\ev}^{t}-\ev^{t}_b \rangle-\mu \langle A({\ev}^{t+1}-\ev^{t+1}_b),A({\ev}^{t}-\ev^{t}_b) \rangle.\label{eq:et-etp1-etb-etp1}
\end{align}
Therefore, applying the Cauchy-Schwarz inequality and the triangle inequality, it follows that
\begin{align}
&\left|\left(\langle{\ev}^{t+1},{\ev}^{t} \rangle-\mu \langle A{\ev}^{t+1},A{\ev}^{t} \rangle\right) -\left(\langle{\ev}_b^{t+1},{\ev}_b^{t} \rangle-\mu \langle A{\ev}_b^{t+1},A{\ev}_b^{t} \rangle\right)\right|\nonumber\\
&\;\;\;\leq|\langle{\ev}^{t+1}-\ev^{t+1}_b,{\ev}^{t}-\ev^{t}_b \rangle|+\mu|\langle A({\ev}^{t+1}-\ev^{t+1}_b),A({\ev}^{t}-\ev^{t}_b) \rangle|\nonumber\\
&\;\;\;\leq (1+\mu (\sigma_{\max}(A))^2) \|{\ev}^{t+1}-\ev^{t+1}_b\|\|{\ev}^{t}-\ev^{t}_b \|.
\end{align}
Define event $\Ec_1$ as
\[
\Ec_1\triangleq \{\sigma_{\max}(A) \le 2\sqrt{m}+\sqrt{n}\}.
\]
As mentioned earlier, 
\[
P(\Ec_1^c)\leq {\rm e}^{-\frac{m}{2}}.
\]
 Hence, conditioned on $\Ec_1$, 
\begin{align}
&\left|\left(\langle{\ev}^{t+1},{\ev}^{t} \rangle-\mu \langle A{\ev}^{t+1},A{\ev}^{t} \rangle\right) -\left(\langle{\ev}_b^{t+1},{\ev}_b^{t} \rangle-\mu \langle A{\ev}_b^{t+1},A{\ev}_b^{t} \rangle\right)\right|\nonumber\\
&\;\;\;\leq \Big(1+ \mu m \left(\sqrt{n\over m}+2\right)^2\Big) \|{\ev}^{t+1}-\ev^{t+1}_b\|\|{\ev}^{t}-\ev^{t}_b \|\nonumber\\
&\;\;\;\leq {4k\over n} \Big(1+ \left(\sqrt{n\over m}+2\right)^2\Big) L^2\delta^{2\alpha},\label{eq:dif-inner-prod-quantized}
\end{align}
where the last line follows from \eqref{eq:dif-evt-evbt} and because $\mu={1\over m}$. 

Next, we bound the quantized term $\langle{\ev}_b^{t+1},{\ev}_b^{t} \rangle-\mu \langle A{\ev}_b^{t+1},A{\ev}_b^{t}\rangle$. To do this, define the set of normalized error vectors as 
\begin{align}
\Fc_b \triangleq \left\{{\xvh-g(\uv)\over \|\xvh-g(\uv)\|}:\;\uv\in\Uc_b^k \right\}.
\end{align}
Clearly, $|\Fc_b|\leq |\Uc_b|^k$. Define event $\Ec_1$ as  
\begin{align}
\Ec_2 \triangleq \left\{ \langle \ev_b,\ev_b'\rangle-{1\over m}\langle A\ev_b,A\ev_b'\rangle \leq 0.45: \forall \;(\ev_b,\ev_b')\in\Fc_b^2 \right\}.
\end{align}
Applying Lemma \ref{lemma:u-v-As-Av-concent}, and the union bound, it follows that
\begin{align}
\P(\Ec_2^c)&\leq |\Uc_b|^{2k}2^{-0.05m}\nonumber\\
&{\leq}  2^{2bk-0.05m}\nonumber\\
&\stackrel{\rm (a)}{\leq}  2^{2k(1+\alpha)\log{1\over \delta} -0.05m}\nonumber\\
&\stackrel{\rm (b)}{\leq} 2^{-2k\upsilon \log{1\over \delta} },
\end{align}
where $\rm (a)$ and $\rm (b)$ hold because $b$, defined in \eqref{eq:def-b}, is smaller than  $\alpha\log {1\over \delta}+1$ and $m$ is greater than  $k40(1+\alpha+\upsilon)\log {1\over \delta} $ by assumption, respectively. 
Finally, conditioned on $\Ec_1\cap\Ec_2$, combining \eqref{eq:et-etp1-etb-etp1} and \eqref{eq:dif-inner-prod-quantized}, it follows that 
\begin{align}
2\left(\langle{\ev}^{t+1},{\ev}^{t} \rangle-\mu \langle A{\ev}^{t+1},A{\ev}^{t} \rangle \right) \|\xvh-\xvh^{t}\|\leq (0.9+\eta) \|\xvh-\xvh^{t}\|,\label{eq:1st-term}
\end{align} 
where $\eta$ is defined in \eqref{eq:def-eta}.

\item $2\mu(\sigma_{\max}(A))^2\|\xv-{\xvh}\|$: 
Note that
\begin{align}
\|\xvh-\xv\| &=\min_{\uv\in\Uc_b^k} \|g(\uv)-\xv\| \leq  \|g([\uvt]_b)-\xv\| \nonumber\\
&= \|g([\uvt]_b)-g(\uvt)+g(\uvt)-\xv\|\nonumber\\
&\leq \|g([\uvt]_b)-g(\uvt)\|+\|g(\uvt)-\xv\|\nonumber\\
&\leq L\|[\uvt]_b-\uvt\|+\sqrt{n}\delta\nonumber\\
&\leq L\sqrt{k}2^{-b}+\sqrt{n}\delta. \label{eq:error-xvh-xv}
\end{align}
Therefore, using \eqref{eq:error-xvh-xv}, conditioned on $\Ec_2$, we have 
\begin{align}
2\mu(\sigma_{\max}(A))^2\|\xv-{\xvh}\|&\leq \left(2+\sqrt{n\over m}\;\right)^2\left( L\sqrt{k}2^{-b}+\sqrt{n}\delta\right)\nonumber\\
&\leq \left(2+\sqrt{n\over m}\;\right)^2\left( L\delta^{\alpha}\sqrt{k\over n}+1\right)\sqrt{n}\delta\nonumber\\
&= \gamma_1 \delta \sqrt{n},\label{eq:2nd-term}
\end{align} 
where $\gamma_1$ is defined \eqref{eq:def-gamma1-thm2}.
\item $2\mu \left|\langle A^T\zv, \ev^{t+1}\rangle\right|$:
First, note that $\langle A^T\zv, \ev^{t+1}\rangle=\langle \zv,A \ev^{t+1}\rangle$, and
 \begin{align}
 |\langle A^T\zv, \ev^{t+1}\rangle|&=|\langle \zv,A \ev^{t+1}\rangle|=|\langle \zv,A (\ev^{t+1}-\ev_b^{t+1}+\ev_b^{t+1})\rangle|\nonumber\\
 &\stackrel{\rm (a)}{\leq} |\langle \zv,A\ev_b^{t+1}\rangle|+|\langle \zv,A (\ev^{t+1}-\ev_b^{t+1})\rangle|\nonumber\\
  &\stackrel{\rm (b)}{\leq} |\langle \zv,A\ev_b^{t+1}\rangle|+\sigma_{\max}(A) \|\zv\| \|\ev^{t+1}-\ev_b^{t+1}\|\nonumber\\
    &\stackrel{\rm (c)}{\leq} |\langle \zv,A\ev_b^{t+1}\rangle|+\sigma_{\max}(A) \|\zv\| L\delta^{\alpha}\sqrt{k\over n},\label{eq:3rd-term}
 \end{align}
 where $\rm (a)$, $\rm (b)$ and $\rm (c)$ follow from the triangle inequality, the Cauchy-Schwarz inequality and \eqref{eq:dif-evt-evbt}, respectively.  Next, to bound $ |\langle \zv,A\ev_b^{t+1}\rangle|$, we employ Lemma \ref{lemma:gaussian-vectors}. For $\tau>0$ and $\tau_z>0$, define events $\Ec_3$ and $\Ec_4$ as 
\begin{align}
\Ec_3\triangleq \{|\langle \zv,A \ev_b  \rangle|\leq \sigma\sqrt{(1+\tau)m}: \ev_b\in\Fc_b\},
\end{align}
and
 \begin{align}
\Ec_4\triangleq \{\|\zv\| \leq \sigma\sqrt{m(1+\tau_z)}\;\},
\end{align}
respectively. By the law of total probability,
\begin{align}
\P(\Ec_3^c)&=\P(\Ec_3^c\cap \Ec_4)+\P(\Ec_3^c\cap \Ec_4^c)\nonumber\\
&\leq \P(\Ec_3^c\cap \Ec_4)+\P( \Ec_4^c).\label{eq:P-Ec3-total-prob}
\end{align}
 For a fixed $\ev_b\in\Fc_b$, $A\ev_b$ is  i.i.d.~$\Nc(0,1)$ and independent of $\zv$. Therefore, by Lemma \ref{lemma:gaussian-vectors}, $\langle \zv,A \ev_b  \rangle$ has the same distribution as $\|\zv\| G_{\ev_b}$, where $G_{\ev_b}$ is independent of $\zv$ and is distributed as $\Nc(0,1)$.    Hence, for a fixed $\ev_b$, 
 \begin{align}
 \P( \langle \zv,A\ev_b^{t+1}\rangle\geq \sigma\sqrt{(1+\tau)m},\Ec_4)&=  \P\left( G_{\ev_b}\|\zv\| \geq  \sigma\sqrt{(1+\tau)m }, \Ec_4\;\right)\nonumber\\
&\leq   \P\left( G_{\ev_b} \geq \sqrt{1+\tau\over 1+\tau_z},\Ec_4^c\;\right)\nonumber\\
&\leq   \P\left( G_{\ev_b} \geq \sqrt{1+\tau\over 1+\tau_z}\;\right)\nonumber\\
&\leq \ex^{-{1+\tau\over 2(1+\tau_z)}},
 \end{align}
 where the last line holds because for $G\sim\Nc(0,1)$ and $\tau>0$,  $\P(G>\tau)\leq \ex^{-\tau^2/2}$.
 Therefore, applying the union bound, it follows that
  \begin{align}
 \P(\Ec_3^c\cap \Ec_4)&\leq  2^{2kb} \ex^{-{1+\tau\over 2(1+\tau_z)}}\nonumber\\
 &\leq  2^{2k(1+(1+\alpha) \log{1\over \delta})} \ex^{-{1+\tau\over 2(1+\tau_z)}}.\label{eq:P-Ec3c-cap-Ec4}
 \end{align}
 Also, by Lemma \ref{chi-squared},
 \begin{align*}
 \P(\Ec_4^c)\leq \ex^{- \frac{m}{2}(\tau_z - \ln(1 + \tau_z))}.
 \end{align*}
 Let $\tau_z=1$. Then, $\tau_z - \ln(1 + \tau_z)>0.3$ and 
 \begin{align}
 \P(\Ec_4^c)\leq \ex^{- 0.15m}.
 \end{align}
 Choosing 
 \[
 \tau=-1+6(1+\alpha)\left(\log{1\over \delta}\right)k,
 \]
the exponent  of the RHS of \eqref{eq:P-Ec3c-cap-Ec4} can be bounded  as follows
 \begin{align}
 {2(\ln2)k(1+(1+\alpha) \log{1\over \delta})}-{1+\tau\over 2(1+\tau_z)}&= {2(\ln2)k(1+(1+\alpha) \log{1\over \delta})}-1.5(1+\alpha)\left(\log{1\over \delta}\right)k\nonumber\\
 &\leq -0.1 (1+\alpha)\left(\log{1\over \delta}\right)k+2(\ln 2)k.
 \end{align}
Therefore,
 \begin{align}
 \P(\Ec_3^c\cap \Ec_4)&\leq   \ex^{-0.1 (1+\alpha)\left(\log{1\over \delta}\right)k+2(\ln 2)k}.\label{eq:P-Ec3c-cap-Ec4-final}
 \end{align}
 Moreover, for this choice of $\tau$, conditioned on $\Ec_3$,
\begin{align}
\mu \langle \zv,A\ev_b^{t+1}\rangle\leq \sigma\sqrt{1+\tau\over m}= \sigma\sqrt{6(1+\alpha)\left(\log{1\over \delta}\right)k\over m}.
\end{align} 
Also, conditioned on $\Ec_1\cap \Ec_3\cap \Ec_3$, 
\begin{align}
{\mu}\sigma_{\max}(A) \|\zv\| L\delta^{\alpha}\sqrt{k\over n} &\leq {2\sqrt{m}+\sqrt{n}\over  m} \sigma \sqrt{2m} L\delta^{\alpha} \sqrt{k\over n}\nonumber\\
&=\sigma \sqrt{2k\over n}\left(2+\sqrt{n\over m}\;\right)L\delta^{\alpha}\nonumber\\
&=\gamma_2\sigma L\delta^{\alpha},
\end{align} 
where $\gamma_2$ is defined in \eqref{eq:def-gamma2-thm2}. Hence, in summary, conditioned on $\Ec_1\cap \Ec_3\cap \Ec_3$,
\begin{align}
{2\mu \over \sigma }\left|\langle A^T\zv, \ev^{t+1}\rangle\right|\leq \sqrt{6(1+\alpha)\left(\log{1\over \delta}\right)k\over m}+ {\gamma_2  L\delta^{\alpha}} \label{eq:3rd-term-final}
\end{align}

\end{itemize}

Having the bounds on the three terms, combining \eqref{eq:1st-term}, \eqref{eq:2nd-term} and \eqref{eq:3rd-term-final}, conditioned on $\Ec_1\cap\Ec_3\cap\Ec_3\cap\Ec_4$, the desired result  follows from dividing both sides of  \eqref{eq:step-2} by ${1\over \sqrt{n}}$.

\section{Conclusions}\label{sec:conclusion}

In this paper, we have theoretically  studied the performance of an idealized CS recovery method that employs  exhaustive search over all the outputs of a GF corresponding to our  desired class of signals $\Qc$. In the asymptotic regime, where  $n$ (the ambient dimension of set $\Qc$) grows without bound, having a family of GFs with input dimension $k=k_n$ and representation error $\delta=\delta_n$ converging to zero, we have shown that, roughly, $k$ measurements are sufficient  for almost lossless  recovery. We have also studied the performance of an efficient algorithm  based on PGD that employs an AE at each iteration to project the updated signal onto the set of desired signals.   We refer to this method as AE-PGD and prove that given enough measurements, the algorithm converges to the vicinity of the optimal solution even in the presence of  additive white Gaussian noise. We have provided simulation results that highlight both the power and the potential weaknesses of  such recovery methods  based on GFs.

\bibliographystyle{unsrt}
\bibliography{myrefs.bib,refs_list.bib}

\end{document}